\documentclass[11pt,letterpaper,english,preprintnumbers,amsmath,amssymb,superscriptaddress,nofootinbib,prl]{revtex4}
\usepackage{mathpazo}
\usepackage[letterpaper]{geometry}
\usepackage[T1]{fontenc}
\usepackage[latin9]{inputenc}
\usepackage{verbatim}
\usepackage{amsmath}
\usepackage{amsthm}
\usepackage{amssymb}
\usepackage{graphicx}
\usepackage{setspace}

\makeatletter

\pdfpageheight\paperheight
\pdfpagewidth\paperwidth

\@ifundefined{textcolor}{}
{%
 \definecolor{BLACK}{gray}{0}
 \definecolor{WHITE}{gray}{1}
 \definecolor{RED}{rgb}{1,0,0}
 \definecolor{GREEN}{rgb}{0,1,0}
 \definecolor{BLUE}{rgb}{0,0,1}
 \definecolor{CYAN}{cmyk}{1,0,0,0}
 \definecolor{MAGENTA}{cmyk}{0,1,0,0}
 \definecolor{YELLOW}{cmyk}{0,0,1,0}
}
 \theoremstyle{definition}
 \newtheorem*{defn*}{\protect\definitionname}
  \theoremstyle{definition}
 \newtheorem{defn}{\protect\definitionname}
  \theoremstyle{plain}
   \newtheorem{lem}{\protect\definitionname}
  \theoremstyle{lemma}
     \newtheorem*{lem*}{\protect\definitionname}
  \theoremstyle{lemma}
  
\theoremstyle{plain}
\newtheorem{thm}{\protect\theoremname}
  \theoremstyle{plain}
  
  \theoremstyle{remark}
  \newtheorem{rem}{\protect\remarkname}

\makeatother

\usepackage{babel}
  \providecommand{\assumptionname}{Assumption}
  \providecommand{\definitionname}{Definition}
  \providecommand{\remarkname}{Remark}
\providecommand{\corollaryname}{Corollary}
\providecommand{\theoremname}{Theorem}

\begin{document}

\title{Power law violation of the area law in quantum spin chains \footnote{PNAS title: "Supercritical entanglement in local systems: Counterexample to the area law for quantum matter"}}

\author{Ramis Movassagh}
\email{q.eigenman@gmail.com}
\affiliation{Department of Mathematics, IBM TJ Watson Research Center, Yorktown
Heights, NY, 10598}

\author{Peter W. Shor}
\email{shor@math.mit.edu}
\affiliation{Department of Mathematics, Massachusetts Institute of Technology,
Cambridge, MA, 02139, USA}

\begin{abstract}
{The sub-volume scaling of the entanglement entropy with the system's size, $n$, has been a subject of vigorous study in the last decade. The area law provably holds for gapped one dimensional systems and it was believed to be violated by at most a factor of $\log\left(n\right)$ in physically reasonable models such as critical systems. 
In this paper, we generalize the spin$-1$ model of Bravyi et al (PRL 2012) to all integer spin-$s$ chains, whereby we introduce a class of exactly solvable models that are physical, yet violate the area law by a power law. The proposed Hamiltonian is local and translationally invariant in the bulk. We prove that it is frustration free and has a unique ground state. Moreover, we prove that the energy gap scales as $n^{-c}$, where using the theory of Brownian excursions, we prove $c\ge2$. This rules out the possibility of these models being described by a relativistic conformal field theory. We analytically show that the Schmidt rank grows exponentially with $n$ and that the half-chain entanglement entropy to the leading order scales as $\sqrt{n}$ (Eq. 1). Geometrically, the ground state is seen as a uniform superposition of all $s-$colored Motzkin walks. Lastly, we introduce an external field which allows us to remove the boundary terms yet retain the desired properties of the model. Our techniques for obtaining the asymptotic form of the entanglement entropy, the gap upper bound and the self-contained expositions of the combinatorial techniques, more akin to lattice paths, may be of independent interest.}
\end{abstract}
\maketitle

Study of quantum many-body systems (QMBS) is the study of quantum properties of matter and quantum resources (e.g., entanglement) provided by matter for building revolutionary new technologies such as a quantum computer. One of the properties of the QMBS is the amount of entanglement among parts of the system \cite{tura2014detecting,NC}. Entanglement can be used as a resource for quantum
technologies and information processing \cite{NC,kimble2008quantum,acin2007entanglement,brun2006correcting}; however, at a
fundamental level it provides information about the quantum state
of matter, such as near-criticality \cite{QCriticality_Coleman2005,osterloh2002scaling}.
Moreover, systems with high entanglement are usually
hard to simulate on a classical computer \cite{Eisert2010}.
How much entanglement do natural QMBSs posses? What are the fundamental limits on simulation of physical systems? 

The area law says that entanglement entropy between two subsystems
of a system is proportional to the area of the boundary between them.
A generic state does not obey an area law \cite{hayden2006aspects};
therefore obeying an area law implies that a QMBS contains much less
quantum correlation than generically expected. One can imagine that
any given system has inherent constraints such as underlying symmetries
and locality of interaction that restrict the states to reside on
special sub-manifolds rendering their simulation efficient \cite{landau2015polynomial}.

Since the discovery that the AKLT model \cite{AKLT} is exactly solvable,
and that the density matrix renormalization group method (DMRG) \cite{white}
works extremely well on one-dimensional (1D) systems, we have come
to believe that 1D systems are typically easy to simulate. The DMRG and
its natural representation by matrix product states (MPS) \cite{cirac}
gave systematic recipes for truncating the Hilbert space based on
ignoring zero and small singular values in specifying the states of
1D systems. DMRG and MPS have been tremendously successful in practice
for capturing the properties of matter in physics and chemistry \cite{kurashige2013entangled,yan2011spin}. We now know that generic local Hamiltonians, unlike the AKLT model, are gapless \cite{movassagh2016generic}. One wonders about the limitations of DMRG.

The rigorous proof of a general area law does not exist; however,
it holds for gapped systems in 1D \cite{Matth_areal}.  In the condensed matter community it is a common belief that gapped
local Hamiltonians of QMBS on a D-dimensional lattice fulfill the
area-law conjecture \cite{Eisert2010}. That is, the entanglement
entropy of a region of diameter $L$ should scale as the area of the
boundary $\mathcal{O}\left(L^{D-1}\right)$ rather than its volume
$\mathcal{O}\left(L^{D}\right)$. In the more general case, when the
ground state is unique but the gap vanishes in the thermodynamical
limit, it is expected that the area-law conjecture still holds, but
now with a possible logarithmic correction, i.e., $S=\mathcal{O}\left(L^{D-1}\log L\right)$
\cite{Eisert2010}. In other words, one expects that as long as the ground state is
unique, the area-law can be violated by at most a logarithmic factor.
In particular, in 1D, it is expected that if we cut a chain of $n$
interacting spins in the middle, the entanglement entropy should scale
at most like $\log n$. This  is based mostly on calculations
done in $1+1$ conformal field theories (CFTs) \cite{Cardy2009,osterloh2002scaling},
as well as, in the Fermi liquid theory \cite{wolf2006violation}.

This belief has been seriously challenged by both quantum information and
condensed matter theorists in recent years. Motivated by QMA-hard
Hamiltonians, there are various interesting examples of 1D Hamiltonian
constructions \cite{MovassaghThesis2012, Irani2010, GottesmannHastings2010}
that can have larger, even linear, scaling of entanglement entropy
with the system's size.
%
In condensed matter physics, non-translationally invariant models
have been proposed and argued to violate the area law maximally (i.e., linearly for a chain) \cite{vitagliano2010volume},
Huijse et al gave a supersymmetric model with some degree of fine-tuning
that violates the area law \cite{huijse2013area}. More recently
Gori et al \cite{gori2014bell} argued that in translationally invariant
models a fractal structure of the fermi surface is necessary for maximum
violation of entanglement entropy, and using non-local field theories
volume-law scaling was argued using simple constructions \cite{shiba2014volume}.
Independently from \cite[Chapter 6]{MovassaghThesis2012} , Ramirez
et al constructed mirror symmetric models satisfying the volume law, i.e., maximum scaling with the system's size possible \cite{ramirez2015entanglement}.
The models described above are all interesting for the intended purposes
but either have very large spins (e.g., $s\ge10$) or involve some degree of fine-tuning. In particular, Irani proposed a  $s=10$ spin-chain model with linear scaling of the  entanglement entropy. This model is translationally invariant, but the local terms depend on the systems' size. This is a fine-tuning, and the spin dimension is quite high \cite{Irani2010}. 

As noted previously, a generic state violates the area
law maximally \cite{hayden2006aspects}. It was largely believed that the ground state of 
``physically reasonable'' models would violate the area law by at
most a $\log\left(n\right)$ factor, where $n$ is the number of particles
\cite[see for a review]{swingle2013universal}. Physically reasonable
models need to have Hamiltonians that are: 1. Local, 2. Translationally
invariant and 3. Have a unique ground state. These requirements, among other things, eliminate highly fine-tuned models. This implies that $\log\left(n\right)$
is the maximum expected entanglement entropy in realistic physical spin chains.

In an earlier work, Bravyi et al \cite{Movassagh2012_brackets} proposed
a spin$-1$ model with the ground state half-chain entanglement entropy $S=\frac{1}{2}\log n+c$,
which is logarithmic factor violation of the area law as expected
during a phase transition. This model is not truly local as it depends
crucially on boundary conditions. The scaling of the entanglement is exactly what one expects for critical systems.

We have found an infinite class of exactly solvable integer spin-$s$ 
 chain models with $s\ge2$ that are physically reasonable and exact calculation of the entanglement
entropy shows that they violate the area law to the leading order
by $\sqrt{n}$ (Eq. \ref{eq:EntanglementEntropy}). The proposed Hamiltonian
is local and translationally invariant in the bulk but the entanglement of the ground state depends on boundary projectors. We prove that it has a unique
ground state and give a new technique for proving the gap that uses universal convergence of random walks to a Brownian motion. We prove that the energy gap scales as $n^{-c}$, where using the theory of
Brownian excursions we show that the constant $c\ge2$. This bound
rules out the possibility of these models being describable by a CFT.
The Schmidt rank of the ground state grows exponentially with
$n$.

\begin{figure}
\centerline{\includegraphics[width=.4\textwidth]{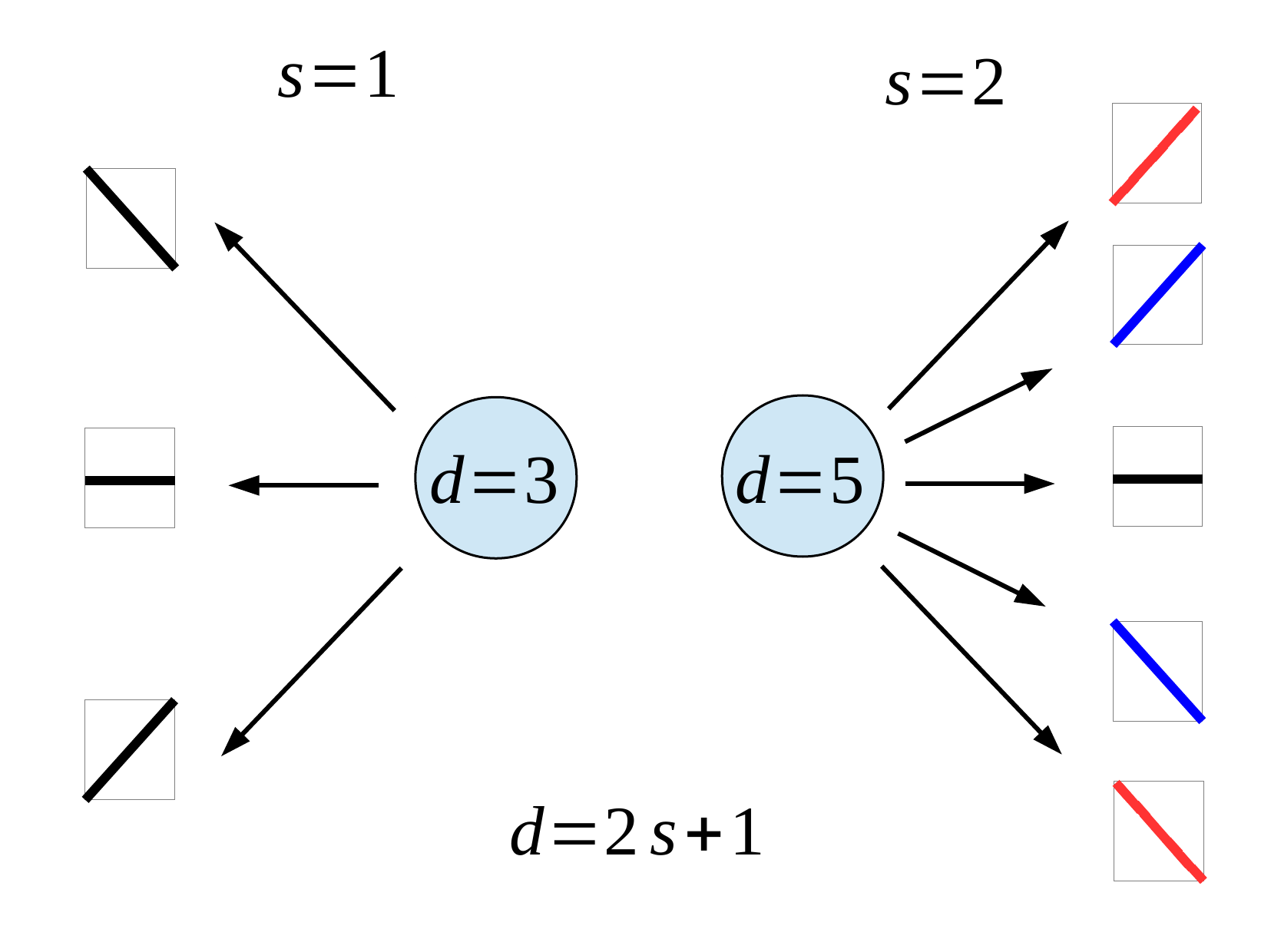}}
\caption{Labeling the states for $s=1$ and $s=2$.}\label{fig:labeling-the-states}
\end{figure}

We then introduce an external field. In presence of the external field the boundary projectors are no longer needed. The model has a \textit{frustrated} ground state, and its gap and entanglement  are solvable. This makes the model  truly local (Eq. \ref{eq:ExternalFieldModel}). We remark that the particle-spins can be as low as $s=2$ for $\sqrt{n}$ violation.
We now describe this class of models and detail the proofs and further discussions in the Supplementary Information (SI).

Let us consider an integer spin$-s$ chain of length $2n$. It is
convenient to label the $d=2s+1$ spin states  as shown in Fig. \ref{fig:labeling-the-states}. Equivalently, and for better readability, we instead use the labels $\left\{u^1,u^2,\cdots,u^s,0,d^1,d^2,\cdots,d^s\right\} $
where $u$ means a step up and $d$  a step down.
We distinguish each \textit{type} of step by associating
a color from the $s$ colors shown as superscripts on $u$ and
$d$.

\begin{figure}
\centerline{\includegraphics[width=.4\textwidth]{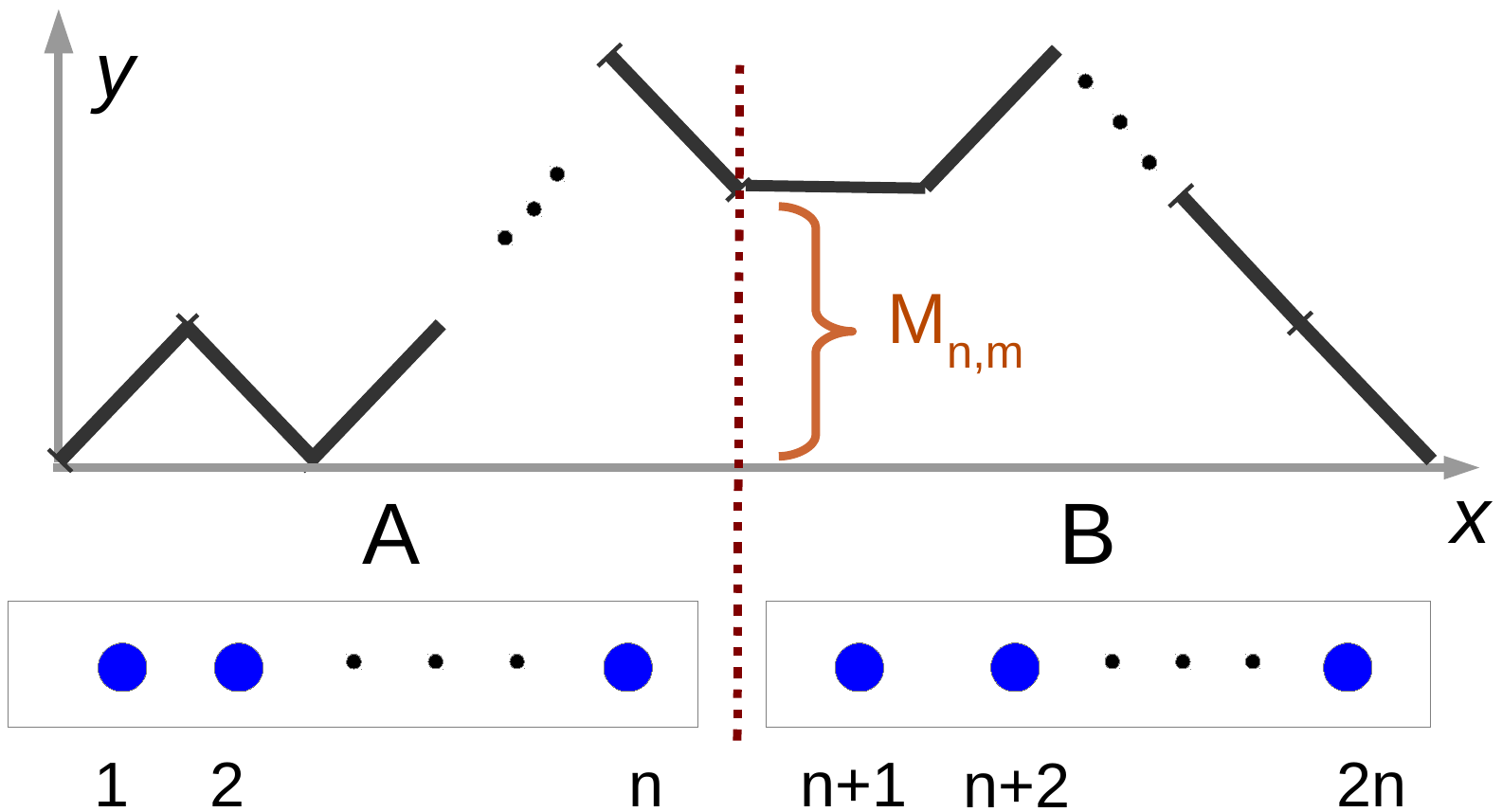}}
\caption{A Motzkin walk of length $2n$ with $s=1$.
There are $M_{n,m}^{2}$ such walks with height $m$ in the middle and coordinates $\left(x,y\right)$:$\left(0,0\right),\left(n,m\right),\left(2n,0\right)$ }\label{fig:Motzkin-walks-Uncolored}
\end{figure}
\begin{figure}
\begin{center}
\centerline{\includegraphics[width=.45\textwidth]{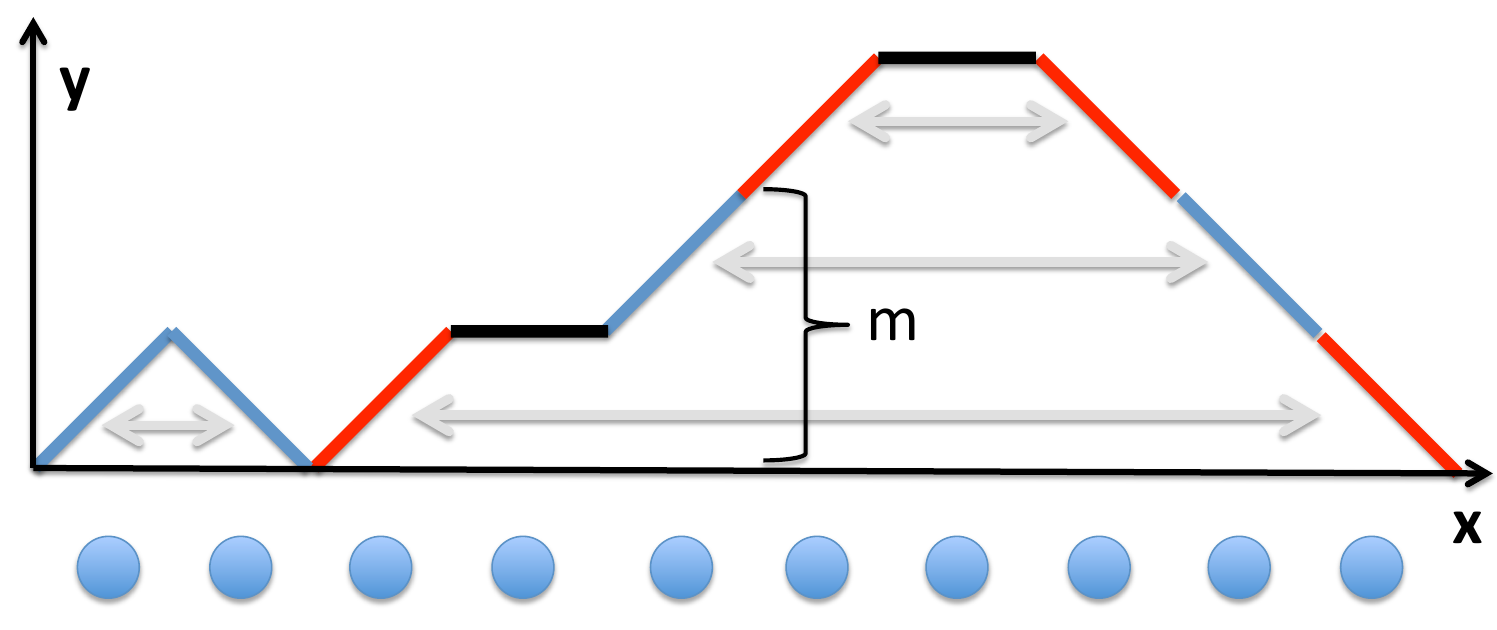}}
\caption{A Motzkin walk with $s=2$ colors of length $2n=10$. The height $m$ quantifies the degree of correlation
between the two halves.}\label{fig:Motzkin-walks-on}
\end{center}
\end{figure}
A Motzkin walk on $2n$ steps is any walk from $\left(x,y\right)=\left(0,0\right)$
to $\left(x,y\right)=\left(2n,0\right)$ with steps $\left(1,0\right)$,
$\left(1,1\right)$ and $\left(1,-1\right)$ that never passes below
the x-axis, i.e., $y\ge0$. An example of such a walk is shown in
Fig. \ref{fig:Motzkin-walks-Uncolored}. The height at the midpoint is $0\le m\le n$
which results from $m$ steps up with the balancing steps down on
the second half of the chain. In our model the unique ground state
is the $s-$colored Motzkin state which is defined to be the uniform
superposition of all $s$ colorings of Motzkin walks on $2n$ steps.
The nonzero heights in the middle are the source of the mutual information
between the two halves and the large entanglement entropy of the half-chain
(Fig. \ref{fig:Motzkin-walks-on}).

The Schmidt rank is $\frac{s^{n+1}-1}{s-1}\approx\frac{s^{n+1}}{s-1}$,
and using a two dimensional saddle point method, the half-chain entanglement
entropy asymptotically is (please see SI for details)
\begin{eqnarray}
S =  2\log_{2}\left(s\right)\sqrt{\frac{2\sigma n}{\pi}}+\frac{1}{2}\log_{2}\left(2\pi\sigma n\right)+(\gamma-\frac{1}{2})\log_{2}e\quad\mbox{bits}\label{eq:EntanglementEntropy}
\end{eqnarray}
where $\sigma=\frac{\sqrt{s}}{2\sqrt{s}+1}$ is constant and $\gamma$ is the
Euler constant. The ground state is a pure state (which we call the
Motzkin state), whose von Neumann entropy is zero. However, the entanglement
entropy quantifies the amount of disorder produced (i.e., information
lost) by ignoring half of the chain. The leading order $\sqrt{n}$
scaling of the entropy establishes that there is a large amount of
quantum correlation between the two halves. 

Consider the following local operations to any Motzkin walk: interchanging
zero with a non-flat step (i.e., $0d^{k}\leftrightarrow d^{k}0$ or
$0u^k\leftrightarrow u^{k}0$) or interchanging a consecutive
pair of zeros with a peak of a given color (i.e., $00\leftrightarrow u^{k}d^{k}$).
Any $s-$colored Motzkin walk can be obtained from another one by
a sequence of these local changes. To construct a local Hamiltonian
with projectors as interactions that has the uniform superposition
of the Motzkin walks as its zero energy ground state, each of the
local terms of the Hamiltonian has to annihilate states that are symmetric
under these interchanges. Local projectors as interactions have the
advantage of being robust against certain perturbations \cite{verstraete2009quantum}.
This is important from a practical point of view and experimental
realizations.

Therefore, the local Hamiltonian, with projectors as interactions,
that has the Motzkin state as its unique zero energy ground state
is
\begin{equation}
H=\Pi_{boundary}+\sum_{j=1}^{2n-1}\Pi_{j,j+1}+\sum_{j=1}^{2n-1}\Pi_{j,j+1}^{cross},\label{eq:H}
\end{equation}
where $\Pi_{j,j+1}$ implements the local changes discussed above
and is defined by 
\begin{equation}
\Pi_{j,j+1}\equiv\sum_{k=1}^{s}\left[|D^{k}\rangle_{j,j+1}\langle D^{k}|+|U^{k}\rangle_{j,j+1}\langle U^{k}|+|\varphi^{k}\rangle_{j,j+1}\langle\varphi^{k}|\right]
\end{equation}
with $|D^{k}\rangle=\frac{1}{\sqrt{2}}\left[|0d^{k}\rangle-|d^{k}0\rangle\right]$, $|U^{k}\rangle=\frac{1}{\sqrt{2}}\left[|0u^{k}\rangle-|u^{k}0\rangle\right]$
and $|\varphi^{k}\rangle=\frac{1}{\sqrt{2}}\left[|00\rangle-|u^{k}d^{k}\rangle\right]$. The
projectors $\Pi_{boundary}\equiv\sum_{k=1}^{s}\left[|d^{k}\rangle_{1}\langle d^{k}|+|u^{k}\rangle_{2n}\langle u^{k}|\right]$
select out the Motzkin state by excluding all walks that start and
end at non-zero heights. Lastly, $\Pi_{j,j+1}^{cross}\equiv \sum_{k\ne i}| u^{k}d^{i}\rangle_{j,j+1}\langle u^{k}d^{i}|$
ensures that balancing is well ordered. For example, we want to ensure that the unbalanced sequence of steps $u^3 u^1 u^2$ is balanced by $d^2 d^1 d^3 $ and not say $d^1 d^3 d^2 $. $\Pi_{j,j+1}^{cross}$ penalizes wrong ordering by prohibiting $00\leftrightarrow u^{k}d^{i}$
when $k\ne i$. These projectors are required only when $s>1$ and
do not appear in \cite{Movassagh2012_brackets}.

The difference between the ground state energy and the energy of the
first excited state is called the gap. One says a system is gapped
when the difference between the two smallest energies is at least
a fixed constant in the thermodynamical limit ($n\rightarrow\infty$).
Otherwise the system is gapless.

Whether a system is gapped has important implications for its physics. When it is gapless, the scaling by which
the gap vanishes as a function of the the system's size, has important
consequences for its physics. For example, gapped systems have exponentially
decaying correlation functions \cite{GottesmannHastings2010}, and
quantum critical systems are necessarily gapless \cite{sachdev2007quantum}.
Moreover, systems that obey a CFT are gapless but the gap must vanish
as $1/n$ \cite{mathieu1997conformal}. Therefore, to quantify the
physics, it is desirable to find new techniques for analyzing the
gap that can be applied in other scenarios.

The model proposed here is gapless and the gap scales as $n^{-c}$
where $c\ge2$ is a constant. We prove this by finding two functions
both of which are inverse powers of $n$ such that the gap is always
smaller than one of them (called an upper bound) and greater than
the other (called a lower bound). We utilize techniques from mathematics such as Brownian excursions and universal convergence of random walks to a Brownian motion, as well as, other ideas from computer science such as linear programming and fractional matching theory.

To prove an upper bound on the gap one needs a state $|\phi\rangle$
that has a small constant overlap with the ground state and such that
$\langle\phi|H|\phi\rangle\ge\mathcal{O}\left(n^{-2}\right)$. Take
\begin{equation}
|\phi\rangle=\frac{1}{\sqrt{M_{2n}}}\sum_{m_{p}}e^{2\pi i\tilde{\theta}\tilde{A}_{p}}\mbox{ }|m_{p}\rangle,\label{eq:phi}
\end{equation}
where the sum is over all Motzkin walks, $M_{2n}$ is the total number
of Motzkin walks on $2n$ steps, $\tilde{A}_{p}$ is the area under
the Motzkin walk $m_{p}$ and $\tilde{\theta}$ is a constant to be
determined by the condition of a small constant overlap with the ground
state. The overlap with the ground state is defined by $\langle{\cal M}_{2n}|\phi\rangle=\left(1/M_{2n}\right)\sum_{m_{p}}e^{2\pi i\tilde{\theta}\tilde{A}_{p}}$.
As $n\rightarrow\infty$, the random walk converges to a Wiener process
\cite{prokhorov1956convergence} and a random Motzkin walk converges
to a Brownian excursion \cite{durrett1977functionals}. We scale
the walks such that they take place on the unit interval. The scaled
area is denoted by $A$ and $\tilde{\theta}\rightarrow\theta$. In
this limit, the overlap becomes (see Fig. \ref{fig:Density_of_Area} for the density and Fig. \ref{fftDensity} for its Fourier transform)
 \footnote{$F_{A}\left(\theta\right)$ is the Fourier transform of the probability
density function which is called the characteristic function.} 
\begin{figure}
\centerline{\includegraphics[width=.4\textwidth]{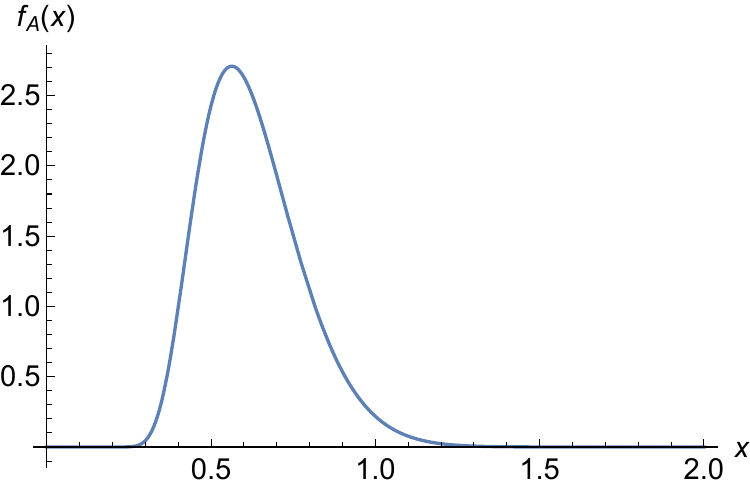}}
\caption{Plot of the probability density of the area under a Brownian excursion $f_{A}\left(x\right)$ on $\left[0,1\right]$.}\label{fig:Density_of_Area}
\end{figure}
\begin{equation}
\lim_{n\rightarrow\infty}\langle{\cal M}_{2n}|\phi\rangle\approx F_{A}\left(\theta\right)\equiv\int_{0}^{\infty}f_{A}\left(x\right)e^{2\pi ix\theta}dx\quad,\label{eq:overlapIntegral}
\end{equation}
where $f_{A}\left(x\right)$ is the probability density function for
the area of the Brownian excursion \cite{janson2007brownian} shown
in Fig. \ref{fig:Density_of_Area}. In Eq. \ref{eq:overlapIntegral},
taking $\theta\ll\mathcal{O}\left(1\right)$, gives $\lim_{n\rightarrow\infty}\langle{\cal M}_{2n}|\phi\rangle\approx1$
because it becomes the integral of a probability distribution. However,
taking $\theta\gg\mathcal{O}\left(1\right)$ gives a highly oscillatory
integrand that nearly vanishes. To have a small constant overlap with
the ground state, we take $\theta$ to be the standard of deviation
of $f_{A}\left(x\right)$. Direct calculation then gives $\langle\phi|H|\phi\rangle=\mathcal{O}\left(n^{-2}\right)$. See SI for details.
This upper bound decisively excludes the possibility of the model
being describable by a conformal field theory \cite{Cardy2009}.
\begin{figure}
\centerline{\includegraphics[width=.4\textwidth]{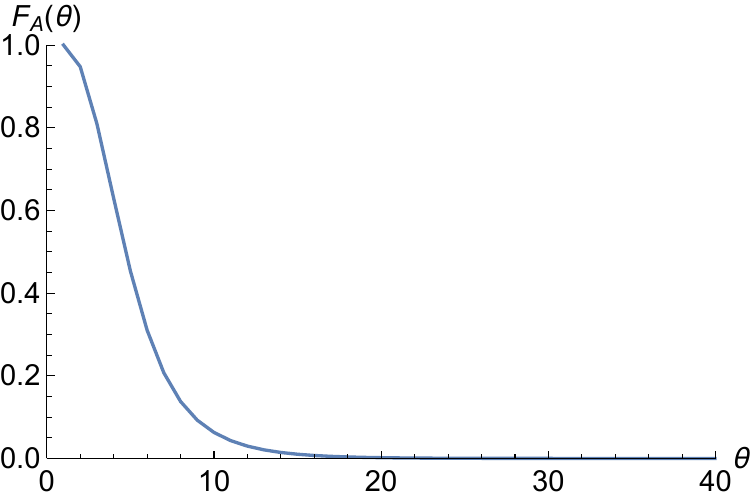}}
\caption{Fourier transform of $f_{A}\left(x\right)$ as defined by Eq. \ref{eq:overlapIntegral}.}\label{fftDensity}
\end{figure}
Using various ideas in perturbation theory, computer science, and
mixing times of Markov chains we obtain a lower bound on the gap that
scales as $n^{-c}$, where $c\gg1$. Since it might be of independent interest in other contexts, we present
a combinatorial and self-contained exposition of the proof in the
SI, different in some aspects from that given in \cite{Movassagh2012_brackets}.

The model above has a unique ground state because the boundary terms
select out the Motzkin state among all other walks with different
fixed initial and final heights. Without the boundary projectors,
all walks that start at height $m_{1}$ and end at height $m_{2}$
with $-2n\le m_{1},m_{2}\le2n$ are ground states. For example, when
$s=1$, the ground state degeneracy grows quadratically with the system's
size $2n$ and exponentially when $s>1$.

For the $s=1$ case, if we impose periodic boundary conditions, then
the the superposition of all walks with an excess of $k$ up (down)
steps is a ground state. This gives $4n+1$ degeneracy of the
ground state, which include unentangled product states.

When $s>1$, each one of the walks with $k$ excess up (down) steps
can be colored exponentially many ways; however, generically they
will \textit{not} be product states. Consider an infinite chain $\left(-\infty,\infty\right)$
and take $s>1$. There is a ground state of this system that corresponds
to the balanced state, where on average for each color,
the state contains as many $u^{i}$ as $d^{i}$. Suppose we restrict
our attention to any block of $n$ consecutive spins. This block contains
the sites $j,j+1,\ldots,j+n-1$, which is a section of a random walk.
Let us assume that it has initial height $m_{j}$ and final height
$m_{j+n-1}$. Further, let us assume that the minimum height of this
section is $m_{k}$ with $j\le k\le j+n-1$. From the theory of random
walks, the expected values of $m_{j}-m_{k}$ and of $m_{j+n-1}-m_{k}$
are $\Theta(\sqrt{n})$. The color and number of any unmatched step ups
 in this block of $n$ spins can be deduced from the remainder
of the infinite walk. Thus a consecutive block of $n$ spins has an
expected entanglement entropy of $\Theta\left(\sqrt{n}\right)$ with
the rest of the chain. A similar argument shows that any block of
$n$ spins has an expected half-block entanglement entropy of $\Theta\left(\sqrt{n}\right)$.

If we take $s=1$, where the ground state can be a product state,
the $\sqrt{n}$ unmatched step up just mentioned can be matched
anywhere on the remaining left and right part of the chain. Two consecutive
blocks of $n$ spins can be unentangled because the number of unbalanced steps
that are matched in the next block is uncorrelated with the number
of unbalanced steps in the first block. However, when $s>1$
the ordering has to match. Even though the number of unbalanced steps
in two consecutive blocks is uncorrelated, the order of the types
of unbalanced steps in them agrees.

The Hamiltonian without the boundary terms is truly translationally
invariant, yet has a degenerate ground state. We now propose a model
with a unique ground state that has the other desirable properties
of the model with boundaries, such as the gap and entanglement entropy
scalings as before. To do so, we put the system in an external field,
where the model is described by the new Hamiltonian 
\begin{eqnarray}
\tilde{H} & \equiv & H+\epsilon\mbox{ }F\label{eq:ExternalFieldModel}\\
F & \equiv & \sum_{i=1}^{2n}\sum_{k=1}^{s}\left(|d^{k}\rangle_{i}\langle d^{k}|\mbox{ }+\mbox{ }|u^{k}\rangle_{i}\langle u^{k}|\right),\nonumber 
\end{eqnarray}
where $H$ is as before but without the boundary projectors and $\epsilon=\epsilon_{0}/n$
with $\epsilon_{0}$ being a small positive constant. It is clear
that $F$ treats $u$ and $d$ symmetrically; therefore, the change
in the energy as a result of applying an external field depends only
on the total number of unbalanced steps denoted by $m$. We denote
the change in the energy of $m$ unbalanced steps by $\Delta E_{m}$.
When $s=1$, the degeneracy after applying the external field will
be, one for the Motzkin state, two-fold when there is a single imbalance,
three-fold for two imbalances, etc. Since the energies are equal for
all $m$ imbalance states, it is enough to calculate the energy for
an excited state with $m$ imbalances resulting only from excess step ups. We denote these states by $|g_{m}\rangle$, where $0\le m\le2n$.

The first order energy corrections, obtained from first order degenerate
perturbation theory, are analytically calculated to be (see SI for details)
\begin{equation}
\epsilon\langle g_{m}|F|g_{m}\rangle\approx4\sigma\epsilon n+\frac{m\epsilon}{8\sqrt{s}}\left(\frac{m}{n}\right)\label{eq:Final_2n}
\end{equation}
The physical conclusion is that the Hamiltonian without the boundary
projectors, in the presence of an external field, $F$, has the Motzkin
state as its unique ground state with energy $4\sigma\epsilon_{0}$.
Moreover, what used to be the rest of the degenerate zero energy states,
acquire energies above $4\sigma\epsilon_{0}$ that for first elementary
excitations scales as $1/n^{2}$. DMRG calculations seem to show that
the actual scaling of the gap, for the system with periodic boundary
conditions in the external field, with the system's size is $n^{-8/3}$ \cite{AdrianChat}. Moreover, the numerical calculations indicate
that the spin-spin correlation functions are flat  \cite{AdrianChat}.
We leave further investigations for future work.

The energy corrections just derived do not mean that the states with
$m$ imbalances will make up for all of the the low energy excitations.
For example, when $s>1$, in the presence of an external field, the
energy of states with a single crossed term will be lower than those
with large $m$ imbalances and no crossings.

Since $||\epsilon F||\ll||H||$, the ground state will deform away
from the Motzkin state to prefer the terms with more zeros in the
superposition. But as long as $\epsilon$ is small, the universality
of Brownian motion guarantees the scaling of the entanglement entropy.
It is, however, not yet clear to us whether $\epsilon$ can be tuned
to a quantum critical point where the ground state has a sharp transition
from highly entangled to nearly a product state. It is possible that
the transition is smooth and that the entanglement continuously diminishes
as $\epsilon$ becomes larger. For example, in the limit where $|\epsilon|\gg||H||/||F||$,
the effective unperturbed Hamiltonian is approximately $F$, whose
ground state is simply the product state $|0\rangle^{\otimes2n}$.

Our model shows that simple physical systems can be much more entangled than expected. From a fundamental physics
perspective, it is surprising that a 1D translationally invariant
quantum spin chain with a unique ground state has about $\sqrt{n}$
entanglement entropy. Moreover, this adds to the collection of exactly
solvable models from which further physics can be extracted. Such
a spin chain can in principle be experimentally realized, and the
large amount of entanglement may be utilized as a resource for quantum technologies and computation. 
\\

\begin{acknowledgments}
We thank Sergey Bravyi and Adrian Feiguin for discussions. RM thanks Herman Goldstine Fellowship
at IBM TJ Watson for the freedom and support, and National Science Foundation (NSF)
for the grant DMS. 1312831. PWS was supported by the US Army
Research Laboratory's Army Research Office through grant number W911NF-12-1-0486,
the NSF through grant number CCF-121-8176, and by the NSF through
the STC for Science of Information under grant number CCF0-939370.
\end{acknowledgments}

%


\newpage
\part*{Supplementary Information: Mathematical Details}

\section{\label{sec:The-Ground-State}The Ground State and its Entanglement}
\subsection{Combinatorics: Dyck paths, Catalan numbers and Motzkin walks}
A Motzkin walk on $2n$ steps is a any
walk, made up of three types of steps: diagonal up, diagonal down
and flat. The walker starts at $\left(x,y\right)=\left(0,0\right)$
and ends at $\left(x,y\right)=\left(2n,0\right)$ such that the walker's
position at any intermediate lattice point $\left(x_{k},y_{k}\right)$
has $y_{k}\ge0$ for all $0\le k\le2n$. The number of all such walks
is counted by the Motzkin numbers $M_{2n}=\sum_{k=0}^{n}\left(\begin{array}{c}
2n\\
2k
\end{array}\right)C_{k}$. Closely related walks are Dyck walks.
\begin{defn}
A Dyck walk (or path) of length $2n$ in the $\left(x,y\right)$ plane
is any path from $\left(0,0\right)$ to $\left(0,2n\right)$ with
steps $\left(1,1\right)$ and $\left(1,-1\right)$, that never passes
below the x-axis. \cite[p. 173]{StanleyVol2} 
\end{defn}

The number of all such walks is counted by the Catalan number $C_{n}\equiv\frac{1}{n+1}\left(\begin{array}{c}
2n\\
n
\end{array}\right)$. 

Catalan numbers are famous numbers in combinatorics. There are hundreds of different combinatorial problems whose solutions are
counted by the Catalan numbers; most of these have been catalogued in \cite{RStanley_CatalanCount}.

We mention in passing that $\lim_{n\rightarrow\infty}C_{n}\approx\frac{4^{n}}{n^{3/2}\sqrt{\pi}}$,
which will be used later to prove the optimality of the proposed canonical
path. 

Now every step up in a Dyck or Motzkin walk has a \textit{unique}
'matching' step down to balance it to ultimately give $y_{2n}=0$.
Suppose there are $s$ colors available, then a Dyck walk of length
$2m$ can be colored $s^{m}$ different ways. For example, take $s=2$
then a walk can be $n$ steps upward with alternating colors of blue
and red which then will uniquely determine the coloring of the remaining
$n$ down steps. 

Any Motzkin walk initially has coordinates $\left(x,y\right)=\left(0,0\right)$.
In the middle of the chain it will have a coordinate $\left(x,y\right)=\left(n,m\right)$
for some $0\le m\le n$; here we denote $m$ to be the "height"
in the middle. To calculate the entropy of a half chain we will need
to count the number of Motzkin walks that start at zero $\left(x,y\right)=\left(0,0\right)$
and reach height $m$ in the middle. A theorem due to André (1887)
counts related (Dyck like) lattice paths \cite[p. 8]{Naryana}.

\begin{thm}
\label{(D.-Andr=0000E9)-Let}(Ballot problem) Let $a,b$ be integers
satisfying $1\leq b\leq a$. The number of lattice paths $\mathcal{N}\left(p\right)$
joining the origin $O$ to the point $\left(a,b\right)$ and not touching
the diagonal $x=y$ except at O is given by 
\begin{equation}
\mathcal{N}\left(p\right)=\frac{a-b}{a+b}\left(\begin{array}{c}
a+b\\
b
\end{array}\right)\label{eq:AndreaEqun}
\end{equation}
 In other words, given a ballot at the end of which candidates $P$,
$Q$ obtain $a$, $b$ votes respectively, the probability that $P$
leads $Q$ throughout the counting of votes is $\frac{a-b}{a+b}$.
\end{thm}

First note that $a=b+1$ gives the Catalan numbers. This theorem can
also be interpreted as counting the number of Dyck walks that reach
a given height for some fixed $x$-coordinate. 

What is the corresponding count of the height of the Motzkin walks
of length $2n$ in the middle (i.e., $x=n$)? Suppose on the half
chain, the Motzkin walk has $k$ zeros. The remaining $n-k$ steps
in this walk are made up of up and down steps. Let the total number
of unmatched step up be $0\le m\le\left(n-k\right)$. Clearly
there are $\left(\begin{array}{c}
n\\
k
\end{array}\right)$ ways to put the zeros and there are $n-k-m$ matched steps.
Hence, there are a total of $\frac{n-k-m}{2}$ matching pairs of steps
on the first $n$ qudits and $m$ unmatched ones, which, for the walk
to be a Motzkin walk, will be matched on the second half of the chain. 

There are $s^{\frac{n-k-m}{2}}$ ways of coloring the matched pairs
on the half chain and $s^{m}$ ways to color the remaining unmatched
up steps. We denote the number of these walks by $M_{n,m,s}$%
\footnote{Not to be confused with the Motzkin numbers $M_{n,s}$.%
}; i.e., the total number of micro-states on the left half chain is
\begin{equation}
\sum_{k=0}^{n-m}\left(\begin{array}{c}
n\\
k
\end{array}\right)s^{\frac{n-k-m}{2}}B_{n-k,m}s^{m}\equiv s^{m}M_{n,m,s}\label{eq:Mm}
\end{equation}
where $B_{n-k,m}$ is the solution of the Ballot problem with height
$m$ on $n-k$ walks. 

Clearly, the Motzkin walk on the second half starts from height $m$
and will eventually reach coordinates $\left(x,y\right)=\left(2n,0\right)$.
Therefore, for every walk on the left half chain that reaches the
height $m$, there are $M_{n,m,s}$ corresponding walks on the right
half that bring it down to zero, i.e., $\left(x,y\right)=\left(2n,0\right)$.
Any choice of coloring of the $m$ unbalanced step ups on the
first half of the chain, uniquely determines the coloring of the second
half. Therefore the total number of $s-$colored Motzkin walks reaching
height $m$ is $s^{m}M_{n,m,s}^{2}$ and the total number of $s-$colored
Motzkin walks of length $2n$ is $N_{n,s}\equiv\sum_{m=0}^{n}s^{m}M_{n,m,s}^{2}$. 

In Eq. \ref{eq:AndreaEqun}, after using $a+b=m-k+1$, $a-b=m+1$
and letting $k\rightarrow n-m-2i$ to take care of parity, 
\begin{eqnarray}
B_{n-k,m} & = & \frac{m+1}{n-k+1}\left(\begin{array}{c}
n-k+1\\
\frac{1}{2}\left(n-k-m\right) 
\end{array}\right)\nonumber \\
&=&\left(\begin{array}{c}
2i+m\\
i
\end{array}\right)-\left(\begin{array}{c}
2i+m\\
i-1
\end{array}\right).\label{eq:Ballot}
\end{eqnarray}

We substitute this into Eq. \ref{eq:Mm}
\begin{eqnarray}
M_{n,m,s} & = & \sum_{i=0}^{\frac{n-m}{2}} s^{i} \left(\begin{array}{c}
n\\
2i+m
\end{array}\right)\left\{ \left(\begin{array}{c}
2i+m\\
i
\end{array}\right)-\left(\begin{array}{c}
2i+m\\
i-1
\end{array}\right)\right\}\nonumber \\
 & = & \left(m+1\right)\sum_{i\ge0}\frac{\left(n\right)! s^{i}}{\left(i+m+1\right)!i!\left(n-2i-m\right)!}\nonumber \\
 & = & \frac{m+1}{n+1}\sum_{i\ge0}s^{i}\left(\begin{array}{ccc}
 & n+1\\
i+m+1 & i & n-2i-m
\end{array}\right)\mbox{ }\nonumber \\
 & \equiv & \sum_{i\ge0}M_{n,m,s,i}.\label{eq:Trinomial} 
\end{eqnarray}
\subsection{The $s-$colored Motzkin state}
\begin{defn}
\label{The-colored-Motzkin}The $s-$colored Motzkin state $|{\cal M}_{2n,s}\rangle$
is the uniform superposition of all $s$ colorings of Motzkin walks
on $2n$ steps defined by 
\[
|{\cal M}_{2n,s}\rangle=\frac{1}{\sqrt{M_{2n,s}}}\sum_{\begin{array}{c}
\mbox{all }s-\mbox{colored}\\
\mbox{Motzkin walks}
\end{array}}|m_{p}\rangle
\]
where $m_{p}$ is an $s-$colored Motzkin walk and $M_{2n,s}$ is
the colored Motzkin number.
\end{defn}
\begin{rem}
For every Motzkin walk reaching height $m$, there are $s^{m}$ eigenvalues
each of size $\frac{M_{n,m}^{2}}{N_{n,s}}$. 
\end{rem}
The Schmidt decomposition of the ground state in the middle of the
chain gives
\begin{equation}
|\mathcal{M}_{2n,s}\rangle=\sum_{m=0}^{n}\sqrt{p_{n,m,s}}\sum_{x\in\left\{ \mbox{ } u^{1},\cdots, u^{s}\mbox{ }\right\} ^{m}}|\hat{C}_{0,m,x}\rangle_{1,\cdots,n}\otimes |\hat{C}_{m,0,\bar{x}}\rangle_{n+1,\cdots,2n}\label{eq:SchDecom}
\end{equation}
where $\hat{C}_{p,q,x}$ is a uniform superposition of all strings
in $\left\{ 0, u^{1},\dots, u^{s}, d^{1},\dots, d^{s}\right\} ^{n}$
with $p$ excess right, $q$ excess step ups and a particular
choice of coloring $x$ of the unmatched steps. For every $x$
there is a unique $\bar{x}$ matching set on the second half of the
chain which is its mirror image. For example if $s=2$, one could
have $x= u^{1} u^{2} u^{2}$ in which case $\bar{x}= d^{2} d^{2} d^{1}$. 
\subsection{Schmidt rank and entanglement entropy}

We now turn to the calculation of the entanglement entropy of the
half chain in the ground state. The Schmidt numbers are
\begin{eqnarray}
p_{n,m,s} & = & \frac{M_{n,m,s}^{2}}{N_{n,s}}\mbox{ },\qquad N_{n,s}\equiv\sum_{m=0}^{n}s^{m}M_{n,m,s}^{2},\label{eq:pAndN_q}
\end{eqnarray}
and the entanglement entropy is given by 
\begin{equation}
S\left(\left\{ p_{n,m,s}\right\} \right)=-\sum_{m=0}^{n}s^{m}p_{n,m,s}\log_{2}p_{n,m,s}.\label{eq:Entropy}
\end{equation}
The \textit{Schmidt rank} is $\frac{s^{n+1}-1}{s-1}\approx\frac{s^{n+1}}{s-1}$
because of the geometric sum on $s^{m}$. 

We are interested in asymptotic scaling of $S\left(\left\{ p_{n,m,s}\right\} \right)$
with the system size. To this end, we shall in what follows, use tools
of asymptotic expansions to evaluate $S\left(\left\{ p_{n,m,s}\right\} \right)$. 

Lets look more carefully at 
\begin{equation}
M_{n,m,s,i}=\left(m+1\right)\left(\begin{array}{ccc}
 & n\\
i+m+1 & i & n-2i-m
\end{array}\right)s^{i}\quad.\label{eq:Trinomial_0}
\end{equation}
If it has a saddle point in the $\left(m,i\right)$-plane, the point
must simultaneous satisfy
\[
\begin{array}{ccccccc}
\frac{M_{n,m,s,i+1}}{M_{n,m,s,i}} & = & 1, & \qquad & \frac{M_{n,m+1,s,i}}{M_{n,m,s,i}} & = & 1\quad.\end{array}
\]

The condition $\frac{M_{n,m,s,i+1}}{M_{n,m,s,i}}=1$ gives $s\left(n-2i-m\right)^{2}-i\left(i+m\right)\approx0$,
yet $\frac{M_{n,m+1,s,i}}{M_{n,m,s,i}}=1$ has its maximum at $m=0$.
In solving for $i$, there are two roots; we choose the one that is
consistent with the $s=1$ result, where $i_{sp}\approx\frac{n}{3}$,
\begin{eqnarray}
i_{sp} & = & \sigma n-\frac{m}{2}+\frac{m}{8\sqrt{s}}\left(\frac{m}{n}\right)+\frac{\left(4s-1\right)m}{128s\sqrt{s}}\left(\frac{m}{n}\right)^{3}+\mathcal{O}\left(n\left(\frac{m}{n}\right)^{5}\right)\label{eq:saddle}\\
 & \approx & \sigma n-\frac{m}{2}+\frac{m}{8\sqrt{s}}\left(\frac{m}{n}\right)\mbox{ },\qquad\sigma\equiv\frac{\sqrt{s}}{2\sqrt{s}+1}\mbox{ }.\nonumber 
\end{eqnarray}

Before getting an asymptotic expansion for Eq. \ref{eq:Trinomial_0},
we consider an example. We will analyze a trinomial coefficient, where
$x+y+z=0$ (noting that $1-2\sigma=\sigma/\sqrt{s}$)
\begin{eqnarray}
\left(\begin{array}{ccc}
 & n\\
\sigma n+x\mbox{ } & \mbox{ }\sigma n+y\mbox{ } & \mbox{ }\left(1-2\sigma\right)n+z
\end{array}\right) & \approx & \sqrt{\frac{2\pi n}{8\pi^{3}\left(\sigma n+x\right)\left(\sigma n+y\right)\left[\left(\frac{\sigma}{\sqrt{s}}\right)n+z\right]}}\label{eq:Trinomial_1}\\
 & \times & \left(\frac{n}{n+x/\sigma}\right)^{\sigma n+x}\left(\frac{n}{n+y/\sigma}\right)^{\sigma n+y}\left(\frac{n}{n+z\sqrt{s}/\sigma}\right)^{\sigma n/\sqrt{s}+z}\nonumber \\
 & \times & \left(\frac{s^{\frac{\sigma}{2\sqrt{s}}}}{\sigma}\right)^{n}s^{\frac{z}{2}}\quad.\nonumber 
\end{eqnarray}
 
But,
\begin{eqnarray*}
\left(\frac{n}{n+x/\sigma}\right)^{\sigma n+x} & = & \exp\left\{ -\left(\sigma n+x\right)\log\left(1+\frac{x}{n\sigma}\right)\right\} \\
 & \approx & \exp\left\{ -\left(\sigma n+x\right)\left(\frac{x}{n\sigma}-\frac{1}{2}\left(\frac{x}{n\sigma}\right)^{2}\right)\right\} \\
 & \approx & \exp\left\{ -x-\frac{x^{2}}{2\sigma n}\right\} 
\end{eqnarray*}
\begin{eqnarray*}
\left(\frac{n}{n+z\sqrt{s}/\sigma}\right)^{\sigma n/\sqrt{s}+z} & = & \exp\left\{ -\left(\frac{\sigma n}{\sqrt{s}}+z\right)\log\left(1+\frac{z\sqrt{s}}{n\sigma}\right)\right\} \\
 & \approx & \exp\left\{ -\left(\frac{\sigma n}{\sqrt{s}}+z\right)\left(\frac{z\sqrt{s}}{n\sigma}-\frac{1}{2}\left(\frac{z\sqrt{s}}{n\sigma}\right)^{2}\right)\right\} \\
 & \approx & \exp\left\{ -z-\frac{z^{2}\sqrt{s}}{2\sigma n}\right\} ;
\end{eqnarray*}
clearly, the expression for $ $$\left(\frac{n}{n+y/\sigma}\right)^{\sigma n+y}\approx\exp\left\{ -y-\frac{y^{2}}{2\sigma n}\right\} $.
In Eq. (\ref{eq:Trinomial_1}), inside the square root is approximately
$\sqrt{\frac{2\pi n}{8\pi^{3}\sigma^{2}\left(\frac{\sigma}{\sqrt{s}}\right)n^{3}}}\approx\frac{s^{1/4}}{2\pi n\sigma^{3/2}}$.
Since $x+y+z=0$,
\begin{equation}
\left(\begin{array}{ccc}
 & n\\
\sigma n+x\mbox{ } & \mbox{ }\sigma n+y\mbox{ } & \mbox{ }\left(1-2\sigma\right)n+z
\end{array}\right)\approx\frac{s^{1/4}}{2\pi n\sigma^{3/2}}\left(\frac{s^{\frac{\sigma}{2\sqrt{s}}}}{\sigma}\right)^{n}s^{\frac{z}{2}}\exp\left(-\frac{x^{2}+y^{2}+\sqrt{s}z^{2}}{2\sigma n}\right)\label{eq:trinomial_example}
\end{equation}

Now we use this result to evaluate Eq. \ref{eq:Trinomial_0} by letting
$i+m=\sigma n+x$ , $i=\sigma n+y$ and $n-2i-m=\left(1-2\sigma\right)n+z$.
Since the standard of deviation of multinomial distributions scales
as $\sqrt{n}$, to get a better asymptotic form, we let $i=i_{sp}+\beta\sqrt{n}$
and $m=\alpha\sqrt{n}$. Hence we identify, 
\begin{eqnarray*}
x & = & \left(\beta+\frac{\alpha}{2}\right)\sqrt{n}+\frac{\alpha^{2}}{8\sqrt{s}}\\
y & = & \left(\beta-\frac{\alpha}{2}\right)\sqrt{n}+\frac{\alpha^{2}}{8\sqrt{s}}\\
z & = & -2\beta\sqrt{n}-\frac{\alpha^{2}}{4\sqrt{s}}
\end{eqnarray*}
 Making these substitutions we get $-\frac{x^{2}+y^{2}+\sqrt{s}z^{2}}{2\sigma n}=-\frac{\alpha^{2}}{4\sigma}-\frac{\sqrt{s}\beta^{2}}{\sigma^{2}}-\mathcal{O}\left(n^{-1/2}\right)$
and $s^{i+\frac{z}{2}}=s^{\sigma n-\frac{\alpha\sqrt{n}}{2}}$ . Therefore,
using Eq. \ref{eq:trinomial_example}, Eq. \ref{eq:Trinomial_0} becomes
\begin{eqnarray*}
M_{n,m,s,i} & = & \frac{\left(m+1\right)}{n+1}\left(\begin{array}{ccc}
 n+1\\
i+m+1 , i , n-2i-m
\end{array}\right)s^{i}
\end{eqnarray*}
This is approximately equal to
\begin{eqnarray*}
M\left(n,s,\alpha,\beta\right) & \equiv & \frac{\left(\alpha\sqrt{n}\right)s^{1/4}}{2\pi n^{2}\sigma^{3/2}}\left(\frac{s^{\frac{\sigma}{2\sqrt{s}}}}{\sigma}\right)^{n}s^{\sigma n-\frac{\alpha\sqrt{n}}{2}} \exp\left(-\frac{\alpha^{2}}{4\sigma}-\frac{\sqrt{s}\beta^{2}}{\sigma^{2}}\right)
\end{eqnarray*}

We need to evaluate $M_{n,m,i,s}$ from $i=0$ to $i=n$. We approximate
this by an integral over $i$ from $0$ to $\infty$. Since $i=\sigma n+\beta\sqrt{n}$,
we have $di=\sqrt{n}\mbox{ }d\beta$. Since the maximum is away from
the boundaries we can extend the integration limit to $-\infty$.
Noting that $\left(\frac{s^{\frac{\sigma}{2\sqrt{s}}}}{\sigma}\right)^{n}s^{i_{sp}}\approx\left(\frac{\sqrt{s}}{\sigma}\right)^{n}s^{-\frac{\alpha\sqrt{n}}{2}}$,
the integration over $\beta$ gives
\begin{eqnarray}
M\left(n,s,\alpha\right) & \equiv & \int d\beta\mbox{ }M\left(n,s,\alpha,\beta\right) =  \frac{\alpha s^{-\alpha\sqrt{n}/2}}{2\sqrt{\pi}n^{3/2}\sigma^{1/2}}\left(\frac{\sqrt{s}}{\sigma}\right)^{n}\exp\left(-\frac{\alpha^{2}}{4\sigma}\right)\label{eq:Mnms}
\end{eqnarray}

Recall that $m=\alpha\sqrt{n}$, hence $s^{m}M_{n,m,s}^{2}$ appearing
in Eq. \ref{eq:Entropy} has an extreme point when 
\begin{equation}
\frac{d}{d\alpha}\alpha^{2}\exp\left[-\frac{\alpha^{2}}{2\sigma}\right]=0\label{eq:alpha_max}
\end{equation}
This happens for $\alpha=\pm\sqrt{2\sigma}$; we clearly need to take
the positive root. For $s=1$, $\alpha=\sqrt{2/3}$ recovers the previous
result \cite{Movassagh2012_brackets}.

Comment: We pause to interpret the nullity of probability at the minimum
$\alpha=m=0$. This corresponds to concatenation of two uniform superpositions
of all Motzkin walks in $n$ steps. Since either half is balanced
by itself, this term is not a source of mutual information between
the two halves and does not contribute to the entanglement entropy. 

We now determine the entropy of the probability distribution $p_{n,m,s}$. After substituting $\alpha=m/\sqrt{n}$ in Eq. \ref{eq:Mnms} and
noting that normalizations cancel, 
\begin{eqnarray}
S\left(\left\{ p_{n,m,s}\right\} \right) & = & -\frac{1}{T}\sum_{m=0}^{n}\frac{m^{2}}{n}\exp\left(-\frac{1}{2\sigma}\frac{m^{2}}{n}\right)\log\left[\frac{1}{T}\frac{m^{2}}{n}s^{-m}\exp\left(-\frac{1}{2\sigma}\frac{m^{2}}{n}\right)\right]\label{eq:EntropyEXACT}\\
T & \equiv & \sum_{m=0}^{n}\frac{m^{2}}{n}\exp\left(-\frac{1}{2\sigma}\frac{m^{2}}{n}\right)\quad.\nonumber 
\end{eqnarray}

We can approximate this with an integral
\begin{eqnarray*}
S\left(\left\{ p_{n,m,s}\right\} \right) & \approx & -\frac{1}{T'}\int_{m=0}^{\infty}dm\mbox{ }\frac{m^{2}}{n}\exp\left(-\frac{1}{2\sigma}\frac{m^{2}}{n}\right)\log\left[\frac{1}{T'}\frac{m^{2}}{n}s^{-m}\exp\left(-\frac{1}{2\sigma}\frac{m^{2}}{n}\right)\right]\\
T' & \approx & \int_{m=0}^{\infty}dm\mbox{ }\frac{m^{2}}{n}\exp\left(-\frac{1}{2\sigma}\frac{m^{2}}{n}\right)\quad.
\end{eqnarray*}
In these integrals we restore the substitution $m=\alpha\sqrt{n}$
to obtain
\begin{eqnarray*}
S\left(\left\{ p_{n,m,s}\right\} \right) & \approx & \frac{1}{2}\log n-\frac{1}{T''}\int_{\alpha=0}^{\infty}d\alpha\mbox{ }\alpha^{2}\exp\left(-\frac{\alpha^{2}}{2\sigma}\right)\log\left[\frac{1}{T''}\alpha^{2}s^{-\alpha\sqrt{n}}\exp\left(-\frac{\alpha^{2}}{2\sigma}\right)\right]\\
T'' & = & \int_{\alpha=0}^{\infty}d\alpha\mbox{ }\alpha^{2}\exp\left(-\frac{\alpha^{2}}{2\sigma}\right)\quad;
\end{eqnarray*}
the factor $\frac{1}{2}\log n$ occurs because $T'=\sqrt{n}T''$ .
Therefore, 
\[
S\left(\left\{ p_{n,m,s}\right\} \right)\approx\frac{1}{2}\log n+\frac{\sqrt{n}}{T''}\log s\int_{\alpha=0}^{\infty}d\alpha\mbox{ }\alpha^{3}\exp\left(-\frac{\alpha^{2}}{2\sigma}\right)-\frac{1}{T''}\int_{\alpha=0}^{\infty}d\alpha\alpha^{2}\exp\left(-\frac{\alpha^{2}}{2\sigma}\right)\log\left[\frac{1}{T''}\alpha^{2}\exp\left(-\frac{\alpha^{2}}{2\sigma}\right)\right]
\]
The two remaining integrals are just numbers and we can calculate
them to obtain the final result
\begin{eqnarray}
S\left(\left\{ p_{n,m,s}\right\} \right) & \approx & 2\log\left(s\right)\mbox{ }\sqrt{\frac{2\sigma}{\pi}}\mbox{ }\sqrt{n}+\frac{1}{2}\log n+\gamma-\frac{1}{2}+\frac{1}{2}\left(\log2+\log\pi+\log\sigma\right)\quad\mbox{nats}\label{eq:H_final}\\
 & = & 2\log_{2}\left(s\right)\mbox{ }\sqrt{\frac{2\sigma}{\pi}}\mbox{ }\sqrt{n}+\frac{1}{2}\log_{2}n+\left(\gamma-\frac{1}{2}\right)\log_{2}e+\frac{1}{2}\left(1+\log_{2}\pi+\log_{2}\sigma\right)\quad\mbox{bits}.\nonumber 
\end{eqnarray}
where $\gamma$ is the Euler gamma number. Note that $s=1$ exactly
recovers the previous result in \cite{Movassagh2012_brackets}. 
\begin{figure}
\centerline{\includegraphics[width=.4\textwidth]{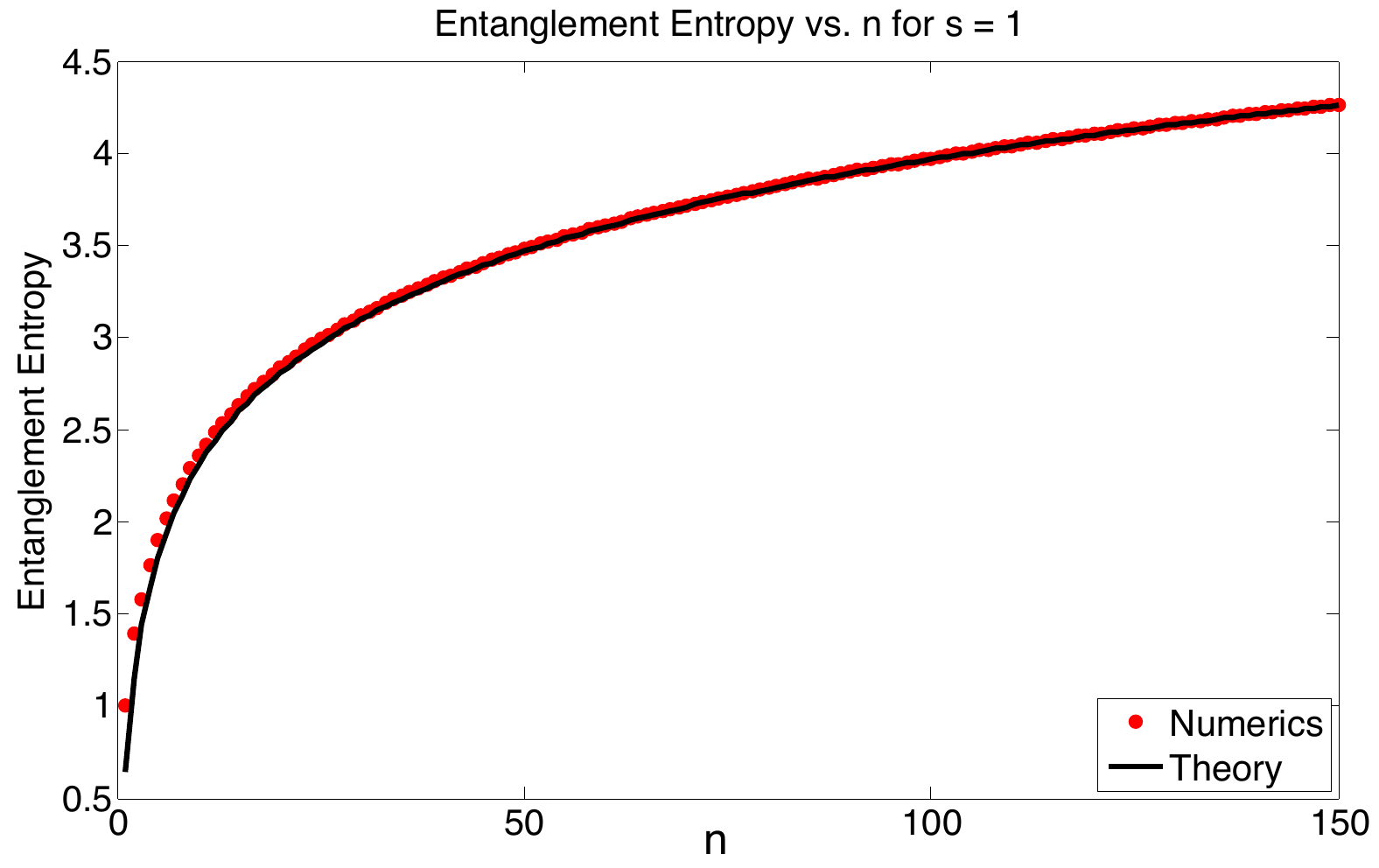}}
\caption{Logarithmic scaling of entanglement entropy for $s=1$.
The red dots are the exact sum (give by Eqs. \ref{eq:Trinomial},\ref{eq:pAndN_q}
and \ref{eq:Entropy}) and the black curve is the asymptotic form
(Eq. \ref{eq:H_final}).}\label{fig:s1}
\end{figure}
\begin{figure}[h]
\centering
{\includegraphics[width=0.4\textwidth]{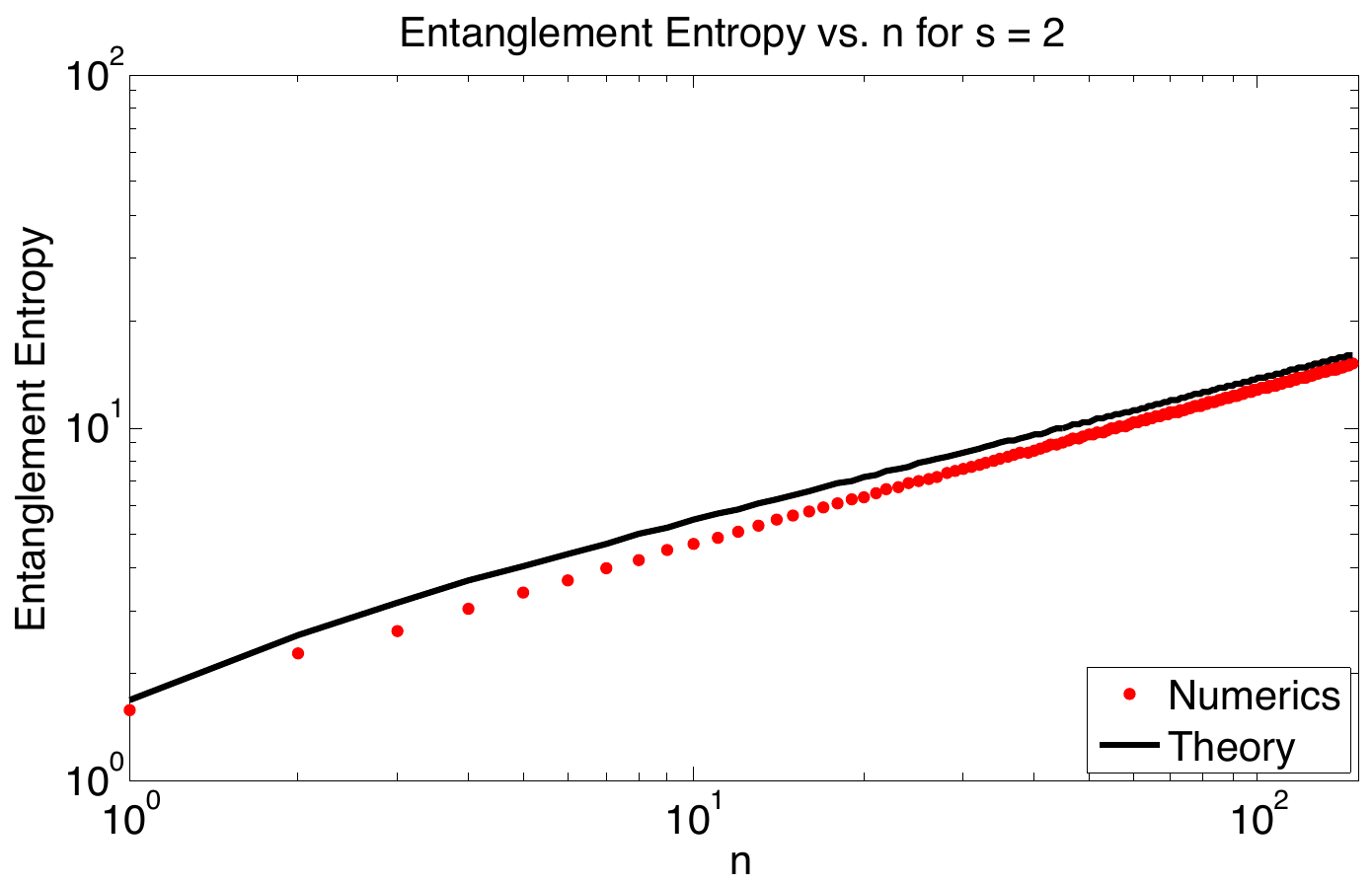}}
\qquad
{\includegraphics[width=0.4\textwidth]{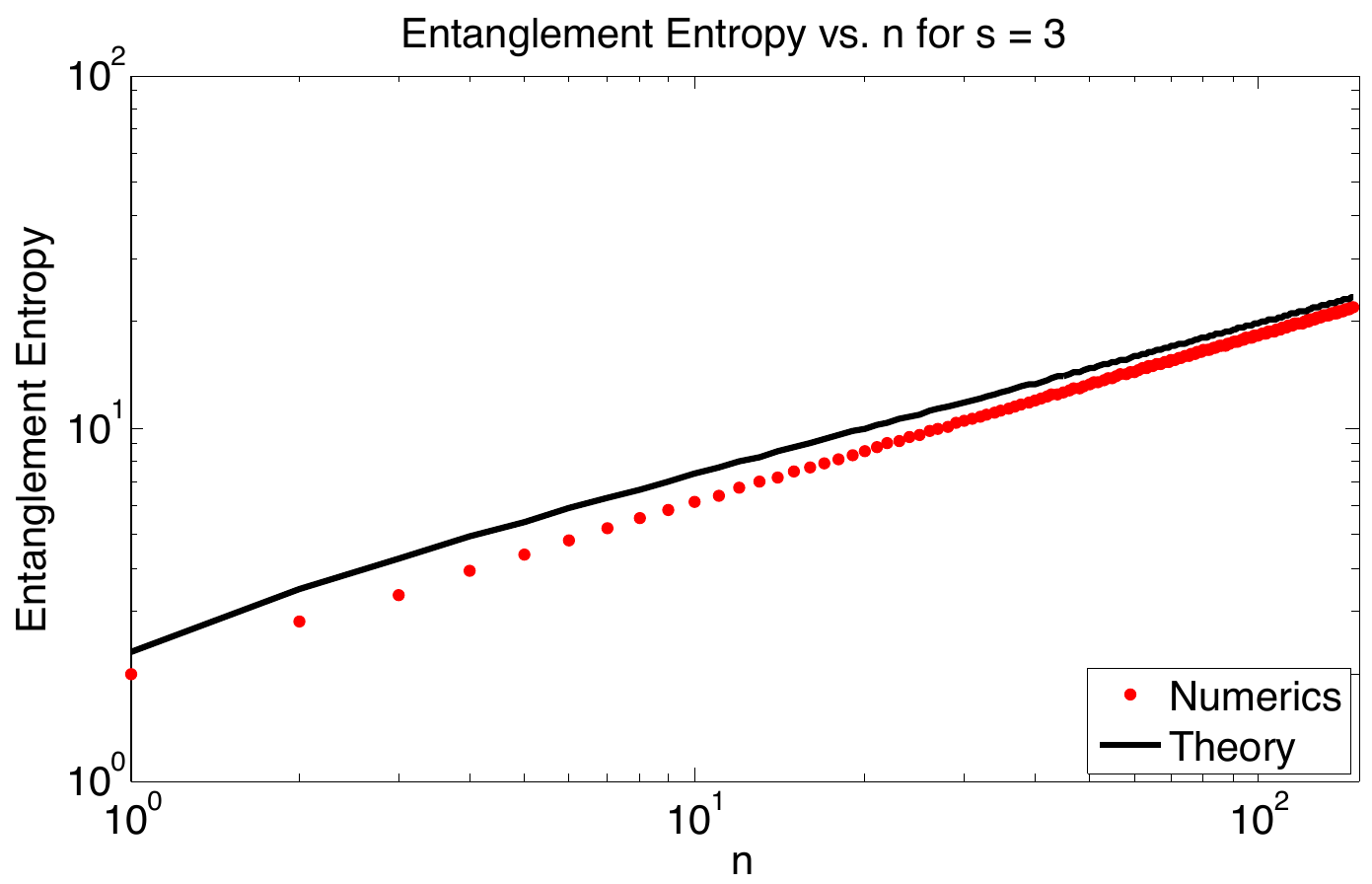}}
{\includegraphics[width=0.4\textwidth]{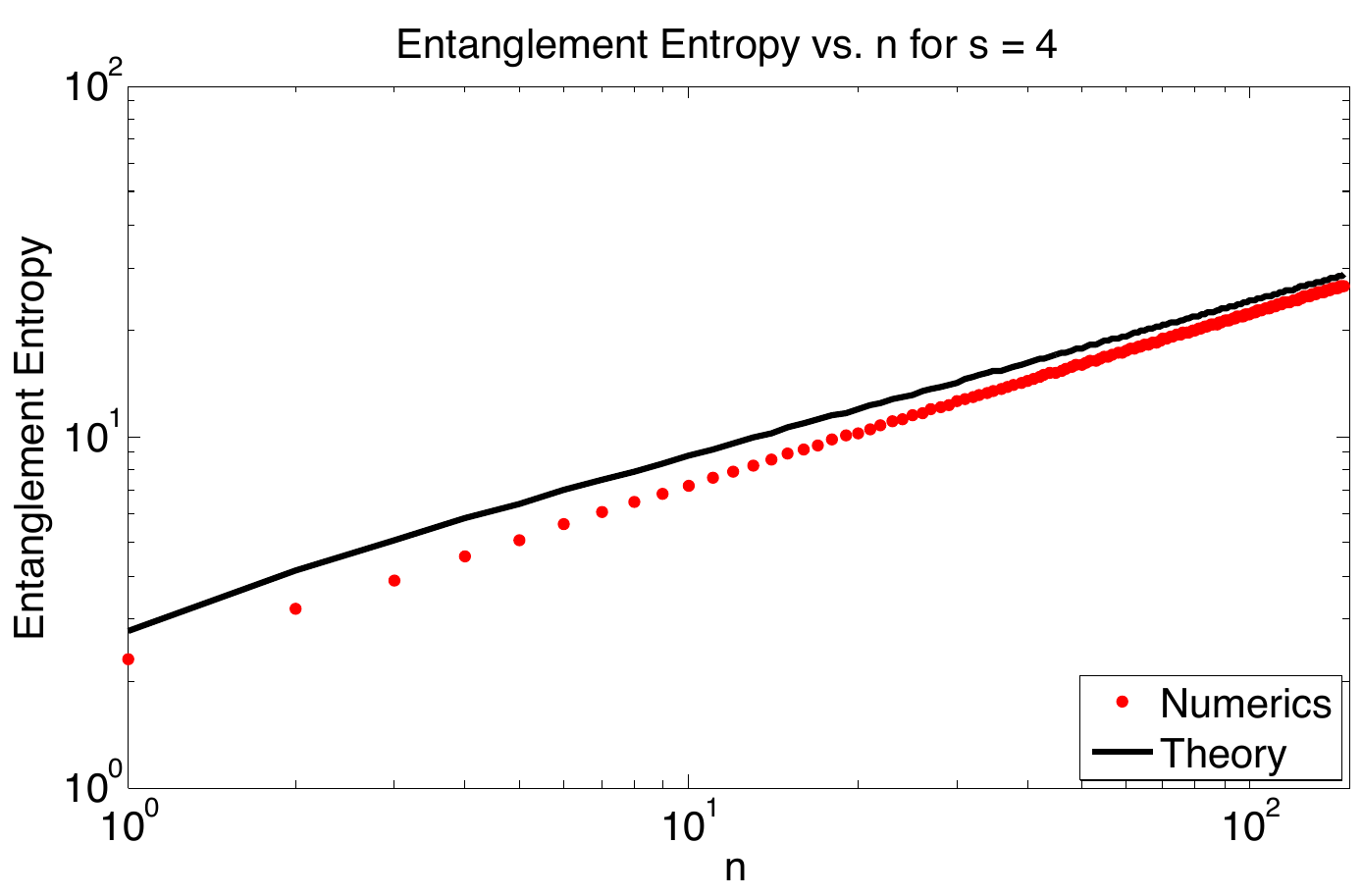}}
\qquad
{\includegraphics[width=0.4\textwidth]{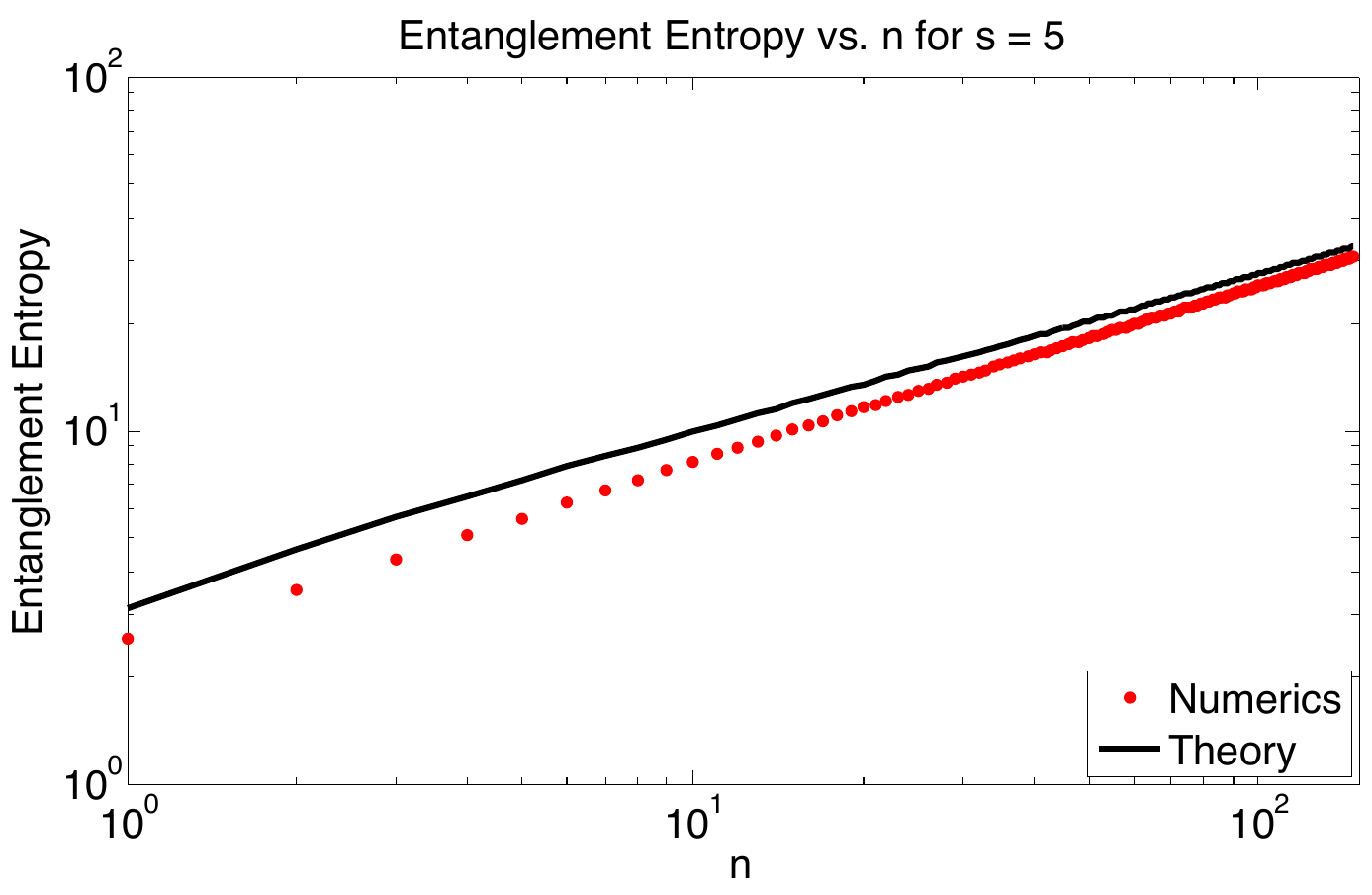}}
\caption{Red dots: the exact sum given by Eqs.\ref{eq:Trinomial},\ref{eq:pAndN_q} and \ref{eq:Entropy}. Black
curve: the asymptotic form (Eq. \ref{eq:H_final}).}
\label{fig:EntanglementEntropy}
\end{figure}

Figs. \ref{fig:s1} and \ref{fig:EntanglementEntropy}, compare the
exact sum given by Eq. (\ref{eq:Entropy} using \ref{eq:Trinomial}
and \ref{eq:pAndN_q}) with the asymptotic result given by Eq.\ref{eq:H_final}.
Fig. \ref{fig:Relative-Error} shows the ratio of the exact sum Eq. (\ref{eq:Entropy}) with Eq. \ref{eq:H_final}.
\begin{figure}
\centerline{\includegraphics[width=.4\textwidth]{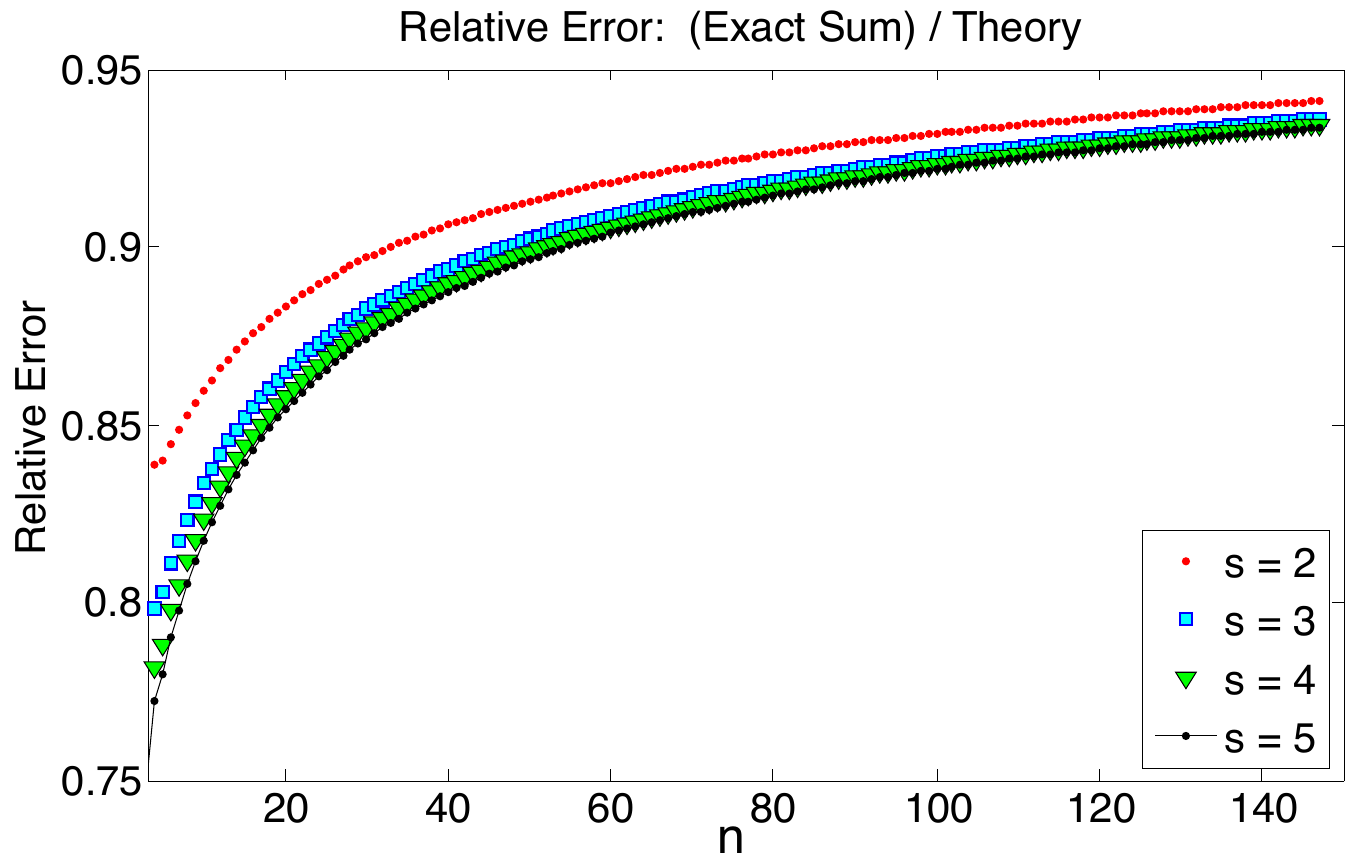}}
\caption{The ratio of Eq. \ref{eq:Entropy} over
Eq. \ref{eq:H_final}; note the asymptotic approach to $1$. }\label{fig:Relative-Error}
\end{figure}
\section{\label{sec:Local-Hamiltonian-and}The Local Hamiltonian and its Gap}
\subsection{The Hamiltonian}
To build a FF local Hamiltonian whose ground state is $|{\cal M}_{2n,s}\rangle$,
we first give a local description of the colored Motzkin walks. As
in \cite{Movassagh2012_brackets}, we say two strings $u$ and $v$
are equivalent, denoted by $u\sim v$, if $u$ can be obtained from
$v$ by a sequence of following local moves:
\begin{eqnarray}
0 d^{k} & \leftrightarrow &  d^{k}0\label{eq:moves}\\
0 u^{k} & \leftrightarrow &  u^{k}0\nonumber \\
00 & \leftrightarrow &  u^{k} d^{k}\quad\forall\mbox{ }k=1,\dots,s\quad;\nonumber 
\end{eqnarray}
these moves can be applied to any consecutive pair of letters. 

Under local moves the matched pairs annihilate one by one and ultimately
the string would have some number of excess unmatched right and/or
step up as well as potentially some crossed pairings. For
example, below are examples of such ``unbalanced'' strings when
$d=5$
\begin{eqnarray*}
 d^{1} d^{2}0 & \cdots & 0 u^{2} u^{1}\mbox{ }\qquad\mbox{unmatched}\\
 u^{1} d^{2}0 & \cdots & 0 u^{2} d^{1}\qquad\mbox{crossed}\\
 d^{2} d^{2}0\cdots0 &  u^{1} d^{2} d^{1} & 0\cdots0 u^{1}\quad\mbox{\mbox{unmatched} and crossed}
\end{eqnarray*}

The equivalent classes of strings that we introduced previously \cite{Movassagh2012_brackets}
are more complicated now, because now strings can get 'jammed' in
various ways under local moves. However, any string except the Motzkin
path will have a minimum nonzero Hamming weight under local moves. 
\begin{lem}
A string $u$ is a Motzkin path iff it is equivalent to the string
of all zeros. 
\end{lem}
\begin{proof}
Under local moves all matched pairs annihilate and there will be no
substrings with $ u^{k} d^{k}$ or $ u^{k}0\cdots0 d^{k}$. If this
is the case, and $u$ contains at least one step up of any
color, then we can focus on the rightmost one and denote it by $ u^{k}$.
We can apply local moves such that there are no zeros to the right
of $ u^{k}$. Then $ u^{k}$ will either be the rightmost letter
of $u$ or it will be followed to its right by an $ d^{i}$ with $i\ne k$.
In either case the Hamming weight of the string is at least $1$.
Similarly if after the local moves, the string $u$ contains at least
one step down, then we can pick the leftmost one and use similar
reasoning to show that under the local moves the minimum Hamming weight
is at least $1$. The only strings that are equivalent to the zero-Hamming
weight string $0^{2n}$ are the colored Motzkin walks.
\end{proof}

We take the ground state to be the uniform superposition of all $s-$colored
Motzkin walks, i.e., strings $u\in\left\{ 0, u^{1}, u^{2},\cdots, u^{s}, d^{1},\cdots, d^{s}\right\} ^{2n}$
that are equivalent to $u_{0}=0^{2n}$. For example, on two qudits,
the ground state is $\sim\left\{ |00\rangle+\sum_{k=1}^{s}| u^{k} d^{k}\rangle\right\} $.

The local Hamiltonian that implements the local moves, given in the
paper, is
\begin{equation}
H=\Pi_{boundary}+\sum_{j=1}^{2n-1}\Pi_{j,j+1}+\sum_{j=1}^{2n-1}\Pi_{j,j+1}^{cross}\label{eq:H}
\end{equation}
with $\Pi_{j,j+1}\equiv\sum_{k=1}^{s} | D^{k}\rangle_{j,j+1}\langle  D^{k}|+|U^{k}\rangle_{j,j+1}\langle U^{k}|+|\varphi^{k}\rangle_{j,j+1}\langle\varphi^{k}|$, where  
\begin{eqnarray}
| D^{k}\rangle=\frac{1}{\sqrt{2}}\left[|0 d^{k}\rangle-| d^{k}0\rangle\right]; \quad |U^k\rangle=\frac{1}{\sqrt{2}}\left[|0 u^{k}\rangle-| u^{k}0\rangle\right];\quad |\varphi^{k}\rangle=\frac{1}{\sqrt{2}}\left[|00\rangle-| u^{k} d^{k}\rangle\right].\label{eq:projectors}
\end{eqnarray}

The projector $| D^{k}\rangle\langle  D^{k}|$ implements $0 d^{k}\leftrightarrow  d^{k}0$
, $|U^k\rangle\langle U^k|$ implements $0 u^{k}\leftrightarrow u^{k}0$
and $|\varphi^{k}\rangle\langle\varphi^{k}|$ implements the \textit{interaction}
term $00\leftrightarrow u^{k} d^{k}$. The last set of projections
$\Pi^{cross}$ penalize wrongly ordered matching of steps of
different types. The rank of the local projectors away from the boundaries
is $s\left(s+2\right)$. The contributions to the rank are $2\left(\begin{array}{c}
s\\
2
\end{array}\right)$ from the penalty terms $\Pi_{j,j+1}^{cross}$, $2s$ from propagation
through the vacuum (i.e., zeros) given by the span of $|U^k\rangle$
and $| D^{k}\rangle$, and $s$ from creation and annihilation of particles
given by the span of $|\varphi\rangle$. 
\begin{lem}
\label{EqualSuperposition}If a state $|\psi\rangle$ is annihilated
by every $\Pi_{j,j+1}$, then it has the same amplitude on any two
strings $u$ and $v$ that are equivalent under the local moves (Eq.
\ref{eq:moves}).
\end{lem}
\begin{proof}

If $ $ $u\sim v$ then there exists a sequence of local moves that
takes $u$ to $v$. Suppose a local move in this sequence is applied
to the $j$ and $j+1$ position of the string to take $u$ to $u'$,
then $\left[\langle u|-\langle u'\right|]_{j,j+1}$ is proportional
to a bra of one of the projectors in Eq. \ref{eq:projectors}. If
$|\psi\rangle$ is annihilated by all the projectors then $\langle u|\psi\rangle=\langle u'|\psi\rangle$
. Since $v$ is obtained from $u$ by a sequence of such local moves
then $\langle u|\psi\rangle=\langle v|\psi\rangle$. 
\end{proof}

The projectors $\Pi_{j,j+1}$ simply move steps through $0$'s
(or vacuum) or create or annihilate a balanced pair of steps of same
color. Under these moves product states, such as $| d^{1}\rangle^{\otimes2n}$,
can also be ground states. Moreover, states such as $| u^{1}\rangle^{\otimes n}\otimes | d^{2}\rangle^{\otimes n}$
are also annihilated by every $\Pi_{j,j+1}$. In order to select out
$|{\cal M}_{2n,s}\rangle$ from all other states, we impose boundary
conditions $\Pi_{boundary}$ that penalize states that are imbalanced
by assigning an energy $1$ to any state that starts with a step down
($ d^{i}$) or ends with a step up ($ u^{i}$) of any color. In addition,
we locally impose $\Pi_{j,j+1}^{cross}$ to prevent crossed pairing
states. 

For example, when $s=2$ the Hamiltonian is
\begin{eqnarray}
H= | d^{1}\rangle_{1}\langle  d^{1}|+| d^{2}\rangle_{1}\langle  d^{2}|+| u^{1}\rangle_{2n}\langle u^{1}|+| u^{2}\rangle_{2n}\langle u^{2}|+  \sum_{j=1}^{2n-1}\Pi_{j,j+1}+\sum_{j=1}^{2n-1}\Pi_{j,j+1}^{cross},\label{eq:H_q=00003D2}
\end{eqnarray}
where $\Pi_{j,j+1}^{cross}=\left\{| u^{1} d^{2}\rangle\langle u^{1} d^{2}|+| u^{2} d^{1}\rangle\langle u^{2} d^{1}|\right\} _{j,j+1}$
and $\Pi_{j,j+1}$ is the span of
\begin{equation}
\begin{array}{ccccc}
| D^{1}\rangle=\frac{1}{\sqrt{2}}\left\{|0 d^{1}\rangle-|  d^{1}0\rangle\right\};\quad | D^{2}\rangle=\frac{1}{\sqrt{2}}\left\{|0 d^{2}\rangle-| d^{2}0\rangle\right\} ; \\
|U^{1}\rangle=\frac{1}{\sqrt{2}}\left\{|0 u^{1}\rangle-| u^{1}0\rangle\right\}; \quad |U^{2}\rangle=\frac{1}{\sqrt{2}}\left\{ |0 u^{2}\rangle-| u^{2}0\rangle\right\}; \\
|\varphi^{1}\rangle=\frac{1}{\sqrt{2}}\left\{ |00\rangle-| u^{1} d^{1}\rangle\right\}; \quad |\varphi^{2}\rangle=\frac{1}{\sqrt{2}}\left\{ |00\rangle-| u^{2} d^{2}\rangle\right\}.
\end{array}\label{eq:LocalProj}
\end{equation}

It follows that the $s-$colored Motzkin state $|{\cal M}_{2n,s}\rangle$,
which is the uniform superposition of all $s-$color Motzkin paths,
is the unique ground state of the FF local Hamiltonian $H$. 
\subsection{Proof of $\mathcal{O}\left(n^{-2}\right)$ upper bound on the gap}
We shall first give a definition of the balanced subspace. We then
assume that we can find a state $|\phi\rangle$ in the balanced subspace,
which has a low energy and a small overlap with the ground state.
We will show that this implies that the first excited state has a
low energy.

\begin{defn}
(balanced subspace) The balanced subspace is the span of the $s-$colored
Motzkin state $|{\cal M}_{2n,s}\rangle$ as defined in Definition \ref{The-colored-Motzkin}. 
\end{defn}

\begin{rem}
In the balanced subspace, the amplitude of any vector on $\Pi_{boundary}$
and $\Pi_{j,j+1}^{cross}$ for all $j$ in the Hamiltonian (Eq. \ref{eq:H})
vanishes. The Hamiltonian then is simply $H=\sum_{j=1}^{2n-1}\Pi_{j,j+1}$.
\end{rem}

For now we take $s=1$ and denote $|{\cal M}_{2n,1}\rangle$ by $|{\cal M}_{2n}\rangle$
and the Motzkin number $M_{2n,1}$ by $M_{2n}$. Take the state $|\phi\rangle$
in the balanced subspace defined by 
\begin{equation}
|\phi\rangle=\alpha_{0}|{\cal M}_{2n}\rangle+\alpha_{1}|e_{1}\rangle+\alpha_{2}|e_{2}\rangle+\cdots+\alpha_{\left(M_{2n}-1\right)}|e_{\left(M_{2n}-1\right)}\rangle\label{eq:AnyState}
\end{equation}
where $|e_{i}\rangle$ is the $i^{th}$ excited states of $H$ in the
balanced subspace. Clearly, $\sum_{i=0}^{{\cal M}_{2n}-1}\left|\alpha_{i}\right|^{2}=1$
and
\[
\langle\phi|H|\phi\rangle=\sum_{i=1}^{{\cal M}_{2n}-1}\left|\alpha_{i}\right|^{2}\langle e_{i}|H|e_{i}\rangle
\]
where we recall that $H=\sum_{j=1}^{2n-1}\left\{ |D\rangle\langle D|+|U\rangle\langle U|+|\varphi\rangle\langle\varphi|\right\} _{j,j+1}$
in the balanced subspace. Let us take $|\phi\rangle$ to have a small
overlap with the ground state, say $\left|\alpha_{0}\right|^{2}\le1/2$.
Then 
\begin{equation}
\langle\phi|H|\phi\rangle\ge\sum_{i=1}^{M_{2n}-1}\left|\alpha_{i}\right|^{2}\langle e_{1}|H|e_{1}\rangle\ge\frac{1}{2}\langle e_{1}|H|e_{1}\rangle\equiv\frac{1}{2}\Delta(H).\label{eq:Lowerbound_Statement}
\end{equation}

We choose $|\phi\rangle$, as defined in the paper, to be $|\phi\rangle\equiv1/\sqrt{M_{2n}}\sum_{m_{p}}e^{2\pi i\tilde{\theta}\tilde{A}_{p}}|m_{p}\rangle,$
where $\tilde{A}_{p}$ is the area under the Motzkin walk $m_{p}$,
$\tilde{\theta}$ is a constant we will specify later and the sum
is over all Motzkin walks. The overlap with the ground state is
\begin{equation}
\langle{\cal M}_{2n}|\phi\rangle=\frac{1}{M_{2n}}\sum_{m=0}^{{\cal M}_{2n}-1}e^{2\pi i\tilde{\theta}\tilde{A}_{p}}.\label{eq:overlap}
\end{equation}

As $n\rightarrow\infty$, the random walk converges to a Wiener process
\cite{prokhorov1956convergence} and a random Motzkin walk converges
to a Brownian excursion \cite{durrett1977functionals}. We wish to
scale the random walk such that it takes place on $\left[0,1\right]$,
which is the standard form and gives $\mathcal{O}\left(1\right)$
mean and variance.

To this end, we first describe the standard Brownian excursion. Let
$B\left(t\right)$ be a standard Brownian motion on $\left[0,1\right]$
with $B\left(0\right)=0$. A standard (normalized) Brownian excursion,
$B_{\mbox{ex}}\left(t\right)$, on the interval $\left[0,1\right]$
is defined by $B\left(t\right)$ conditioned on $B\left(t\right)\ge0$
and $B\left(1\right)=0$ for $t\in\left[0,1\right]$ (see \cite{janson2007brownian}).
Let $B_{\mbox{ex}}\left(t\right)$ be a Brownian excursion and \cite[p. 84]{janson2007brownian}
\[
{\cal B}_{\mbox{ex}}\equiv\int_{0}^{1}B_{\mbox{ex}}\left(t\right)\mbox{ }dt
\]
the Brownian excursion area. The moments of ${\cal B}_{\mbox{ex}}$
are given by 
\[
\mathbb{E}\left[{\cal B}_{\mbox{ex}}^{k}\right]=\frac{4\pi2^{-k/2}k!}{\Gamma\left[\left(3k-1\right)/2\right]}K_{k}\qquad k\ge0,
\]
where $K_{0}=-\frac{1}{2}$ and 
\[
K_{k}=\frac{3k-4}{4}K_{k-1}+\sum_{j=1}^{k-1}K_{j}K_{k-j},\qquad k\ge1.
\]
It follows that $\mathbb{E}\left[{\cal B}_{\mbox{ex}}\right]=4\sqrt{\frac{\pi}{2}}K_{1}=\frac{1}{2}\sqrt{\frac{\pi}{2}}\approx0.626657$
and the standard of deviation $\sigma=\sqrt{5/12-\pi/8}\approx0.1548144$.
In addition $\mathbb{E}\left[{\cal B}_{\mbox{ex}}^{k}\right]\sim3\sqrt{2}k\left(\frac{k}{12e}\right)^{k/2}$
as $k\rightarrow\infty$. 

Let $f_{A}\left(x\right)$ be the probability density function of
${\cal B}_{ex}$. The analytical form of $f_{A}\left(x\right)$ is
\cite[see Eq. 92]{janson2007brownian}
\begin{equation}
f_{A}\left(x\right)=\frac{2\sqrt{6}}{x^{2}}\sum_{j=1}^{\infty}v_{j}^{2/3}\mbox{ }e^{-v_{j}}\mbox{ }U\left(-\frac{5}{6},\frac{4}{3};v_{j}\right)\quad x\in\left[0,\infty\right)\label{eq:f_A}
\end{equation}
with $v_{j}=2\left|a_{j}\right|^{3}/27x^{2}$ where $a_{j}$ are the
zeros of the Airy function, $\mbox{Ai}\left(x\right)$, and $U$ is
the confluent hypergeometric function \cite{tricomi1947sulle}. See
Figs. \ref{fig:Density_of_Area} for the plot of $f_{A}\left(x\right)$.

The total area%
\footnote{From convergence to a Brownian motion, we expect the height in the
middle to be $c'\sqrt{n}$, where $c'$ can be calculated from our
previous techniques. Indeed, the expected height of the Motzkin walk
in the middle (i.e., at site $n$) is
\[
\mbox{\ensuremath{\mathbb{E}}}\left[m\right]=\frac{\sum_{m=0}^{\infty}m\mbox{ }M_{n,m}^{2}}{\sum_{m=0}^{\infty}\mbox{ }M_{n,m}^{2}}
\]
where as before we denote the height by $0\le m\le n$ and $M_{n,m}$
is the number of walks that start from zero and end at height $m$
in $n$ steps. Using similar derivation leading to Eq. \ref{eq:Mnms}
we find
\[
\mbox{\ensuremath{\mathbb{E}}}\left[m\right]\approx\frac{\int_{0}^{\infty}d\alpha\mbox{ }\alpha^{3}\exp\left(-3\alpha^{2}/2\right)}{\int_{0}^{\infty}d\alpha\mbox{ }\alpha^{2}\exp\left(-3\alpha^{2}/2\right)}\mbox{ }\sqrt{n}=2\sqrt{\frac{2}{3\pi}}\sqrt{n}\quad.
\]
} under all strictly positive Motzkin walks of length $n$ satisfies
a recursion relation $A_{n+1}=2A_{n}+3A_{n-1}$ \cite{sulanke2001bijective};
hence $A_{n}=\frac{1}{4}\left(3^{n+1}+\left(-1\right)^{n}\right)$.
By convergence of random walks to Brownian motion $A_{n}$, to the
leading order, is equal to the total area of Motzkin walks that we
are interested in (i.e., non-negative). Since asymptotically $M_{n}=\frac{3^{n}}{n^{3/2}}\sqrt{\frac{27}{4\pi}}\left(1+\mathcal{O}\left(1/n\right)\right)$,
the expected area is 
\[
\mathbb{E}\left(\tilde{A}_{p}\right)=\frac{A_{2n}}{M_{2n}}\approx\sqrt{\frac{2\pi}{3}}n^{3/2}.
\]

We can now solve for the scaling constants. We shall find $c$ such
that $\mathbb{E}\left[\tilde{A}_{p}\right]=cn^{3/2}\mathbb{E}\left[{\cal B}_{\mbox{ex}}\right]$;
therefore, $\sqrt{\frac{2\pi}{3}}=c\frac{1}{2}\sqrt{\frac{\pi}{2}}$
which gives $c=4/\sqrt{3}$. 

We take $\mbox{ }\tilde{A}_{p}=\frac{4}{\sqrt{3}}n^{3/2}x$ and $\tilde{\theta}=\frac{\sqrt{3}}{4}n^{-3/2}\theta$.
With the scaling just performed, most of the probability mass is supported
on $x=\mathcal{O}\left(1\right)$. 

Evaluation of the sum given by Eq. \ref{eq:overlap} in the limit
gives %
\footnote{$F_{A}\left(\theta\right)$ is the Fourier transform of the probability
density function which is called the characteristic function.%
} 
\begin{equation}
\lim_{n\rightarrow\infty}\langle{\cal M}_{2n}|\phi\rangle\approx F_{A}\left(\theta\right)\equiv\int_{0}^{\infty}f_{A}\left(x\right)e^{2\pi ix\theta}dx\quad.\label{eq:overlapIntegral}
\end{equation}

In Eq. \ref{eq:overlapIntegral}, taking $\theta\ll\mathcal{O}\left(1\right)$,
gives $\lim_{n\rightarrow\infty}\langle{\cal M}_{2n}|\phi\rangle\approx1$;
however, $\theta\gg\mathcal{O}\left(1\right)$ gives a highly oscillatory
integrand that nearly vanishes (Fig. \ref{fftDensity}).
To have a small constant overlap with the ground state, we now show,
that $\theta=\mathcal{O}\left(1\right)$. Suppose we choose an area
interval $\left[x_{1},x_{2}\right]$, where $f_{A}\left(x_{1}\right)=f_{A}\left(x_{2}\right)=y$.
We have
\begin{eqnarray*}
\int_{0}^{\infty}f_{A}\left(x\right)e^{2\pi ix\theta}dx & = & \int_{0}^{x_{1}}f_{A}\left(x\right)e^{2\pi ix\theta}dx+\int_{x_{2}}^{\infty}f_{A}\left(x\right)e^{2\pi ix\theta}dx\\
 & + & \int_{x_{1}}^{x_{2}}\left[f_{A}\left(x\right)-y\right]e^{2\pi ix\theta}dx+y\int_{x_{1}}^{x_{2}}e^{2\pi ix\theta}dx\quad.
\end{eqnarray*}

Now if we let $\theta\equiv\frac{1}{x_{2}-x_{1}}$, then $y\int_{x_{1}}^{x_{2}}e^{2\pi ix\theta}dx=0$,
hence 
\begin{eqnarray}
\int_{0}^{\infty}f_{A}\left(x\right)e^{2\pi ix\theta}dx & = & \int_{0}^{x_{1}}f_{A}\left(x\right)e^{2\pi ix\theta}dx+\int_{x_{2}}^{\infty}f_{A}\left(x\right)e^{2\pi ix\theta}dx+\int_{x_{1}}^{x_{2}}\left[f_{A}\left(x\right)-y\right]e^{2\pi ix\theta}dx\nonumber \\
 & \le & \int_{0}^{x_{1}}f_{A}\left(x\right)\mbox{ }dx+\int_{x_{2}}^{\infty}f_{A}\left(x\right)\mbox{ }dx+\int_{x_{1}}^{x_{2}}\left[f_{A}\left(x\right)-y\right]dx\nonumber \\
 & = & 1-y\left(x_{2}-x_{1}\right)\quad.\label{eq:Rectangle}
\end{eqnarray}

We wish to maximize the area of the rectangle $y\left(x_{2}-x_{1}\right)$,
so we take $x_{2}-x_{1}=\sigma$, which makes $y=\mathcal{O}\left(1\right)$
and $\tilde{\theta}=\frac{\sqrt{3}}{2\sqrt{5/3-\pi/2}}n^{-3/2}$. 

It is easy to see that $\sum_{j}\langle\phi|D\rangle_{j,j+1}\langle D|\phi\rangle$
is nonzero if it relates two walks that only differ by
a local move of type $0d\leftrightarrow d0$ at the $j,j+1$ position
\begin{equation}
\sum_{j}\langle\phi|D\rangle_{j,j+1}\langle D|\phi\rangle=\frac{1}{2M_{2n}}\sum_{j}\sum_{m_{p}.m_{q}}e^{2\pi i\tilde{\theta}\left(\tilde{A}_{p}-\tilde{A}_{q}\right)}\langle m_{q}|D\rangle_{j,j+1}\langle D|m_{p}\rangle\label{eq:expectationR}
\end{equation}
The change in the area is either one or zero. There are three types
of nonzero contributions per $j,j+1$ in Eq. \ref{eq:expectationR}:
\begin{eqnarray*}
\left(m_{q}\right)_{j,j+1}=\left(m_{p}\right)_{j,j+1} & : & e^{2\pi i\mbox{ }\tilde{\theta}\left(\tilde{A}_{p}-\tilde{A}_{q}\right)}\langle m_{q}|D\rangle_{j,j+1}\langle D|m_{p}\rangle=1\\
\left(m_{q}\right)_{j,j+1}=0d,\quad\left(m_{p}\right)_{j,j+1}=d0 & : & e^{2\pi i\mbox{ }\tilde{\theta}\left(\tilde{A}_{p}-\tilde{A}_{q}\right)}\langle m_{q}|D\rangle_{j,j+1}\langle D|m_{p}\rangle=-e^{i2\pi\tilde{\theta}}\\
\left(m_{q}\right)_{j,j+1}=d0,\quad\left(m_{p}\right)_{j,j+1}=0d & : & e^{2\pi i\mbox{ }\tilde{\theta}\left(\tilde{A}_{p}-\tilde{A}_{q}\right)}\langle m_{q}|D\rangle_{j,j+1}\langle D|m_{p}\rangle=-e^{-i2\pi\tilde{\theta}}.
\end{eqnarray*}
Note that there is no dependence on the actual values of $\tilde{A}_{p}$
and $\tilde{A}_{q}$ but only on their difference. Putting it together
we find that Eq. \ref{eq:expectationR} gives
\begin{eqnarray*}
\sum_{j}\langle\phi|D\rangle_{j,j+1}\langle D|\phi\rangle=\frac{1}{M_{2n}}\sum_{j}a_{j}\left[1-\cos\left(2\pi\tilde{\theta}\right)\right] \approx\frac{1}{M_{2n}}\sum_{j}2\pi^{2}\tilde{\theta}^{2}a_{j},
\end{eqnarray*}
where $a_{j}$ is the number of strings that have $0d$ or $d0$ at
the position $j,j+1$. An entirely similar calculation gives
\begin{eqnarray*}
\sum_{j}\langle\phi|U\rangle_{j,j+1}\langle U|\phi\rangle & = & \frac{1}{M_{2n}}\sum_{j}b_{j}\left[1-\cos\left(2\pi\tilde{\theta}\right)\right] \approx\frac{1}{M_{2n}}\sum_{j}2\pi^{2}\tilde{\theta}^{2}b_{j},\\
\sum_{j}\langle\phi|\varphi\rangle_{j,j+1}\langle\varphi|\phi\rangle & = & \frac{1}{M_{2n}}\sum_{j}c_{j}\left[1-\cos\left(2\pi\tilde{\theta}\right)\right] \approx\frac{1}{M_{2n}}\sum_{j}2\pi^{2}\tilde{\theta}^{2}c_{j},
\end{eqnarray*}
where $b_{j}$ and $c_{j}$ are the number of strings that have $0u,u0$
and $00,ud $ at positions $j,j+1$ respectively.

Summing up the foregoing equations and using $\tilde{\theta}=\frac{\sqrt{3}}{2\sqrt{5/3-\pi/2}}n^{-3/2}$
we obtain
\begin{equation}
\langle\phi|H|\phi\rangle=\frac{9\pi^{2}}{10-3\pi}\frac{n^{-3}}{M_{2n}}\sum_{j}\left(a_{j}+b_{j}+c_{j}\right).\label{eq:abc}
\end{equation}
 
We need to show that $\frac{a_{j}+b_{j}+c_{j}}{M_{2n}}=\mathcal{O}\left(1\right)$
to get $\mathcal{O}\left(n^{-2}\right)$ upper bound. It is clear
that $c_{j}=M_{2\left(n-1\right)}$ and that $a_{j}\approx b_{j}\approx c_{j}$.
Lastly $1/3\le M_{2\left(n-1\right)}/M_{2n}\le1$; therefore 
\[
\langle\phi|H|\phi\rangle=\mathcal{O}\left(n^{-2}\right).
\]

If we take a general integer $s\ge1$ then using similar reasoning
\begin{equation}
\langle\phi|H|\phi\rangle\sim\frac{2\pi^{2}sn^{-3}}{M_{2n,s}}\sum_{j}\left(a'_{j}+b'_{j}+c'_{j}\right)=\mathcal{O}\left(n^{-2}\right).\label{eq:upperbound_O(n_2)}
\end{equation}
\subsection{Lower bound on the gap of $\mbox{poly}\left(1/n\right)$ in the balanced subspace}
In additional to the 'balanced subspace' above, we define the 'unbalanced
subspace' and summarize the proof idea below before presenting the
formal proof. 

\begin{defn}
(unbalanced subspace) The space orthogonal to the span of $|{\cal M}_{2n,s}\rangle$.
In the unbalanced subspace, the crossings and/or an overall imbalance
can occur.
\end{defn}

The summary of the proof is as follows:

\begin{itemize}
\item Restrict the Hamiltonian to the balanced subspace, where there are
a balanced number of correctly ordered down and up steps
of each color.
\item Identify the terms in the Hamiltonian that implement $0 d^{k}\leftrightarrow  d^{k}0$
and $0 u^{k}\leftrightarrow u^{k}0$ with $H_{move}$. Identify
the interaction terms that implement $00\leftrightarrow u^{k} d^{k}$
with $H_{int}$.
\item The Hamiltonian in the balanced subspace is expressed as $H=H_{move}+H_{int}$.
\item $\Delta\left(H_{move}\right)$ is known, let $H_{\epsilon}\equiv H_{move}+\epsilon H_{int}$
for $0<\epsilon\le1$ and show that $\Delta\left(H_{\epsilon}\right)\le\Delta\left(H\right)$.
\item Use the projection lemma to relate $\Delta\left(H_{\epsilon}\right)$
to $\Delta\left(H_{move}\right)$ and the gap of the restriction of
$H_{int}$ to the ground subspace of $H_{move}$, denoted by $H_{eff}$. 
\item Lower bound the gap of $H_{eff}$ by proving a large spectral gap
of a corresponding Markov chain. We do so by proving rapidly mixing
using the canonical path technique and ideas from fractional matching
in combinatorial optimization.
\item Lastly, lower bound the ground states in the unbalanced subspace.
\end{itemize}

As discussed above in the balanced subspace the Hamiltonian is simply 
\[
H=\sum_{j=1}^{2n-1} \sum_{k=1}^{s}\left[| D^{k}\rangle_{j,j+1}\langle  D^{k}|+|U^k\rangle_{j,j+1}\langle U^k|+|\varphi^{k}\rangle_{j,j+1}\langle\varphi^{k}|\right] ,
\]
where any state automatically vanishes on the boundary terms and the
steps are correctly ordered (i.e., non-crossing).

Let $D_{m}^{s}$ be the set of Dyck paths of length $2m\le2n$ with
$s$ colorings and $D^{s}$ be the union of all $D_{m}^{s}$. Let
$\mathcal{M}_{2n,s}$ be the set of Motzkin paths of length $2n$,
recall that the number of these walks is counted by the colored-Motzkin
number $M_{2n,s}=\left|\mathcal{M}_{2n,s}\right|$. Define a Dyck
space $H_{D}^{s}$ whose basis vectors are Dyck paths $\mathfrak{s}\in D^{s}$.
Given a Motzkin path $u$ with $2m$ steps and any coloring,
let $\mbox{Dyck}\left(u\right)\in D_{m}^{s}$ be the Dyck path obtained
from $u$ by removing zeros. We shall use an embedding $V:\mathcal{H}_{D}^{s}\rightarrow\mathcal{H}_{M}^{s}$
defined by
\begin{equation}
V|\mathfrak{s}\rangle=\frac{1}{\sqrt{\left(\begin{array}{c}
2n\\
2m
\end{array}\right)}}\sum_{\mbox{Dyck}\left(u\right)=\mathfrak{s}}|u\rangle,\quad\begin{array}{c}
\mathfrak{s}\in D^{s}\cap D_{m}^{s}\\
u\in\mathcal{M}_{2n,s}
\end{array}\label{eq:V}
\end{equation}
where $\left(\begin{array}{c}
2n\\
2m
\end{array}\right)$ is the number of ways a given Dyck walk of length $2m$ can be embedded
into Motzkin walks of length $2n$ each having $2\left(n-m\right)$
zeros. It is easily checked that $V^{\dagger}V=I$; i.e., $V$ is
an isometry: 
\[
\langle\mathfrak{t}|V^{\dagger}V|\mathfrak{s}\rangle=\frac{1}{\left(\begin{array}{c}
2n\\
2m
\end{array}\right)}\sum_{\mbox{Dyck}\left(u\right)=\mathfrak{s}}\sum_{\mbox{Dyck}\left(v\right)=\mathfrak{t}}\langle v|u\rangle=\langle\mathfrak{t}|\mathfrak{s}\rangle.
\]
\subsubsection{Perturbation Theory}
Similar to our previous work \cite{Movassagh2012_brackets}, we write
the Hamiltonian restricted to the balanced subspace as $H\equiv H_{move}+H_{int}$,
where
\begin{eqnarray}
H_{move} & \equiv & \sum_{j=1}^{2n-1}\sum_{k=1}^{s}\left[| D^{k}\rangle_{j,j+1}\langle  D^{k|}+|U^k\rangle_{j,j+1}\langle U^{k}|\right]\label{eq:Hmove}\\
H_{int} & \equiv & \sum_{j=1}^{2n-1}\sum_{k=1}^{s}|\varphi^{k}\rangle_{j,j+1}\langle\varphi^{k}|\quad,\label{eq:Hint}
\end{eqnarray}
this notation makes explicit that the steps move through zeros,
or 'vacuum', freely as the local moves indicate (Eq. \ref{eq:moves}).
Yet when a left and a step down of a given color reach one
another they can annihilate to produce a $00$ state; alternatively
a pair of $00$ can spontaneously create a balanced set of steps
in correspondence to the local moves (Eq. \ref{eq:moves}). 

\begin{defn}
Let $\lambda_{1}\left(H\right)$ denote the ground state energy of
$H$ and let $\lambda_{2}\left(H\right)$ denote the second smallest
eigenvalue. We denote the gap of the Hamiltonian $H$ by $\Delta\left(H\right)\equiv\lambda_{2}\left(H\right)-\lambda_{1}\left(H\right)$ 
\end{defn}

Let us consider the interaction term as a perturbation to $H_{move}$
and define a modified Hamiltonian $H_{\epsilon}=H_{move}+\epsilon H_{int}$
for $0<\epsilon\le1$. $H_{\epsilon}$ involves the same projectors,
which means that $|{\cal M}_{2n,s}\rangle$ is a unique ground state
of $H_{\epsilon}$ as well. It is clear that $H\ge H_{\epsilon}$;
therefore $\Delta\left(H\right)\ge\Delta\left(H_{\epsilon}\right)$.

$H_{move}$ annihilates states that are symmetric under the local
moves $0 u^{k}\leftrightarrow u^{k}0$ and $0 d^{k}\leftrightarrow  d^{k}0$.
Therefore, $H_{move}$ coincides with the spin-$1/2$ quantum Heisenberg
chain, whose gap was rigorously calculated to be $ $ $\Delta\left(H_{move}\right)=1-\cos\left(\frac{\pi}{n}\right)=\Omega\left(n^{-2}\right)$
\cite{Koma95}. Moreover, $H_{move}|{\cal M}_{2n,s}\rangle=0$, so
$|{\cal M}_{2n,s}\rangle$ is \textit{a }ground state of $H_{move}$
that has a degenerate ground space. 

We use perturbation theory to compute the $\Delta\left(H_{\epsilon}\right)$
for which we need (see the Lemma below \cite{KempeKitaevRegev}) the
{\it restriction} of $H_{int}$ onto the subspace of the Motzkin
space. The restriction, denoted by%
\footnote{In the projection lemma one defines the restriction by $\tilde{H}_{eff}=\Pi_{move}H_{int}\Pi_{move}$,
where $\Pi_{move}=V\mbox{ }V^{\dagger}$ is the projection onto the
ground subspace of $H_{move}$. Note that $\tilde{H}_{eff}=VH_{eff}V^{\dagger}$
and $H_{eff}=V^{\dagger}\tilde{H}_{eff}V$. Their action is equivalent
with the distinction that the former acts on Motzkin walks and the
latter on Dyck walks. } 
\begin{equation}
H_{eff}\equiv V^{\dagger}H_{int}V,\label{eq:Heff}
\end{equation}
is the process by which we first embed a Dyck walk with a particular
coloring assignment into a Motzkin space (i.e., adding zeros), then
either cut a peak of a given color (i.e., any $ u^{k} d^{k}$) or
add a peak of a given color where there are two consecutive zeros.
Therefore, $H_{eff}$ acts on the Dyck space $H_{D}^{s}$.

\begin{lem}
The unique ground state of $H_{eff}$ is $|D^{s}\rangle$ that satisfies
$|\mathcal{M}_{2n,s}\rangle=V|D^{s}\rangle$ and is given by
\begin{equation}
|D^{s}\rangle=\frac{1}{\sqrt{M_{2n,s}}}\sum_{m=0}^{n}\sqrt{\left(\begin{array}{c}
2n\\
2m
\end{array}\right)}\sum_{\mathfrak{s}\in D_{m}^{s}}|\mathfrak{s}\rangle.\label{eq:groundStateHeff}
\end{equation}
\end{lem}
\begin{proof}
We proved that $|{\cal M}_{2n,s}\rangle$ is the unique ground state
of $H=H_{move}+H_{int}$. Since $H|{\cal M}_{2n,s}\rangle=HV|D^{s}\rangle=0$,
$H_{eff}|D^{s}\rangle=0$. We now prove that $|D^{s}\rangle$ is the
unique ground state. For if it were not, then $H_{eff}|D'\rangle=0$
for a state $|D'\rangle$ other than $|D^{s}\rangle$. But by construction
$H_{move}V|D'\rangle=0$, therefore $\left(H_{move}+H_{int}\right)V|D'\rangle=0$.
Since $|{\cal M}_{2n,s}\rangle$ is the unique ground state of $H$,
then we reach a contradiction unless $|D'\rangle=|D^{s}\rangle$. 
\end{proof}

$H_{eff}$ defines a random Markov process. However, before describing
the Markov process, we shall use the Projection Lemma \cite{KempeKitaevRegev},
that in our notation reads
\begin{lem*}
(Projection Lemma \cite{KempeKitaevRegev}) $H_{eff}$ acts on the
ground subspace of $H_{move}$. If the spectral gap of $H_{move}$
and $H_{eff}$ are both $\mbox{poly}\left(1/n\right)$ then the spectral
gap of $H_{\epsilon}$ is also $\mbox{poly}\left(1/n\right)$ for
small enough $\epsilon>0$. Mathematically, 
\begin{equation}
\epsilon\lambda_{1}\left(H_{eff}\right)-\frac{\mathcal{O}\left(\epsilon^{2}\right)\left\Vert H_{int}\right\Vert ^{2}}{\Delta\left(H_{move}\right)-2\epsilon\left\Vert H_{int}\right\Vert }\le\lambda_{1}\left(H_{\epsilon}\right)\le\epsilon\lambda_{1}\left(H_{eff}\right).\label{eq:ProjectionLemma_GS-1}
\end{equation}
where we take $\Delta\left(H_{move}\right)=\Omega\left(n^{-2}\right)>2\epsilon\left\Vert H_{int}\right\Vert $.
\end{lem*}

Since $H_{move}|\psi\rangle=0$ so long as $|\psi\rangle$ is symmetric
under $0 u^{k}\leftrightarrow u^{k}0$ and $0 d^{k}\leftrightarrow  d^{k}0$,
the ground subspace of $H_{move}$ is spanned by $V|\mathfrak{s}\rangle$
for $|\mathfrak{s}\rangle\in D^{s}$ and is degenerate. Therefore
we can subtract the $\mbox{Span}\left(|{\cal M}_{2n,s}\rangle\langle{\cal M}_{2n,s}|\right)$
from the Hilbert space and consider the Projection Lemma on the orthogonal
complement, whereby the subspace with the smallest eigenvalue becomes
the gap. From the first inequality we have 
\begin{equation}
\Delta\left(H_{\epsilon}\right)\ge\epsilon\Delta\left(H_{eff}\right)-\frac{\mathcal{O}\left(\epsilon^{2}\right)\left\Vert H_{int}\right\Vert ^{2}}{\Delta\left(H_{move}\right)-2\epsilon\left\Vert H_{int}\right\Vert }.\label{eq:ProjectionLemma_gap}
\end{equation}
If we choose an $\epsilon\ll n^{-3}$ then $\Delta\left(H_{move}\right)$
can be considered large with respect to $\epsilon\left\Vert H_{int}\right\Vert $,
which gives
\[
\Delta\left(H\right)\ge\Delta\left(H_{\epsilon}\right)\ge\epsilon\Delta\left(H_{eff}\right)-\mathcal{O}\left(\epsilon^{2}n^{4}\right).
\]

Hence it suffices to prove $\Delta\left(H_{eff}\right)\ge n^{-\mathcal{O}\left(1\right)}.$
\subsubsection{Random walk description}
Let $\pi$ be the induced probability distribution on $D^{s}$
with entries $\pi\left(\mathfrak{s}\right)=\langle\mathfrak{s}|D^{s}\rangle^{2}$. Define the matrix $P$ by
\begin{eqnarray}
P & = & \mathbb{I}-\frac{1}{s\left(2n-1\right)}\mbox{diag}\left(\frac{1}{\sqrt{\pi^{T}}}\right)\mbox{ }H_{eff}\mbox{ }\mbox{diag}\left(\sqrt{\pi}\right)\quad\nonumber \\
P\left(\mathfrak{s,t}\right) & = & \delta_{\mathfrak{s,t}}-\frac{1}{s\left(2n-1\right)}\langle\mathfrak{s}|H_{eff}|\mathfrak{t}\rangle\sqrt{\frac{\pi\left(\mathfrak{t}\right)}{\pi\left(\mathfrak{s}\right)}},\label{eq:P}
\end{eqnarray}
where the second equation explicitly shows the entries.

We claim that $P$ describes a random walk on the set of Dyck paths
$D^{s}$ such that given a pair of Dyck paths $\mathfrak{s,t}\in D^{s}$,
$P\left(\mathfrak{s,t}\right)$ is a transition probability from $\mathfrak{s}$
to $\mathfrak{t}$ and ${\pi}$ is the unique steady state.
One has 
\begin{enumerate}
\item $P$ is stochastic. We use completeness to prove that the sum of any
row is one, i.e., $\sum_{\mathfrak{t}}P\left(\mathfrak{s,t}\right)=1$
\begin{eqnarray*}
\sum_{\mathfrak{t}}\left\{ \delta_{\mathfrak{s,t}}-\frac{\langle\mathfrak{s}|D^{s}\rangle^{-1}}{s\left(2n-1\right)}\langle\mathfrak{s}|V^{\dagger}H_{int}V|\mathfrak{t}\rangle\langle\mathfrak{t}|D^{s}\rangle\right\}  =  \\
1-\frac{\langle\mathfrak{s}|D^{s}\rangle^{-1}}{s\left(2n-1\right)}\langle\mathfrak{s}|V^{\dagger}H_{int}V\sum_{\mathfrak{t}}|\mathfrak{t}\rangle\langle\mathfrak{t}|D^{s}\rangle = \\
1-\frac{\langle\mathfrak{s}|D^{s}\rangle^{-1}}{s\left(2n-1\right)}\langle\mathfrak{s}|V^{\dagger}H_{int}V|D^{s}\rangle = \\
 1-\frac{\langle\mathfrak{s}|D^{s}\rangle^{-1}}{s\left(2n-1\right)}\langle\mathfrak{s}|V^{\dagger}H_{int}|{\cal M}_{2n,s}\rangle =  1 
\end{eqnarray*}
since the Hamiltonian is FF and the colored-Motzkin state is a zero
eigenvector of $H_{int}$. 
\item $P$ has a unique steady state $\pi\left(\mathfrak{s}\right)$ because
$\Sigma_{\mathfrak{s}}\pi\left(\mathfrak{s}\right)P\left(\mathfrak{s,t}\right)=\sum_{\mathfrak{s}}\left\{ \pi\left(\mathfrak{s}\right)\delta_{\mathfrak{s,t}}-\frac{1}{s\left(2n-1\right)}\langle\mathfrak{s}|H_{eff}|\mathfrak{t}\rangle\sqrt{\pi\left(\mathfrak{s}\right)\pi\left(\mathfrak{t}\right)}\right\} =\pi\left(\mathfrak{t}\right)$. 
\item $P$ is reversible, that is $\pi\left(\mathfrak{s}\right)P\left(\mathfrak{s,t}\right)=\pi\left(\mathfrak{t}\right)P\left(\mathfrak{t,s}\right)$
for all $\mathfrak{s,t}$, as can easily be checked (note that $\langle\mathfrak{s}|H_{eff}|\mathfrak{t}\rangle=\langle\mathfrak{t}|H_{eff}|\mathfrak{s}\rangle$). 
\item $P\left(\mathfrak{s,t}\right)=0$ unless $\mathfrak{s}$ and $\mathfrak{t}$
are related by adding or removing a single peak of any color (see
the proof of Lemma \ref{lem:L2}). 
\item Since they are related by a similarity transformation, $\Delta\left(H_{eff}\right)=s\left(2n-1\right)\left(1-\lambda_{2}\left(P\right)\right).$ 
\end{enumerate}
\begin{lem}

\label{lem:L2}$P\left(\mathfrak{s,s}\right)\ge1/2$. Let $\mathfrak{s,t}\in D^{s}$
be any Dyck paths such that $\mathfrak{t}$ can be obtained from $\mathfrak{s}$
by adding or removing a single $ u^{k} d^{k}$ pair of any color
$k$. Then $P\left(\mathfrak{s,t}\right)=\Omega\left(1/n^{3}\right)$.
Otherwise $P\left(\mathfrak{s,t}\right)=0$.
\end{lem}
\begin{proof}

First we prove that if $\mathfrak{s}$ and $\mathfrak{t}$ differ
in more than two consecutive positions then $P\left(\mathfrak{s,t}\right)=0$.
More explicitly let $\mathfrak{s}\ne\mathfrak{t}$ and $|u\rangle$
be a Motzkin path in $V|\mathfrak{t}\rangle$ and $\langle v|$ a Motzkin
path in $\langle\mathfrak{s}|V^{\dagger}$ then 
\begin{equation}
P\left(\mathfrak{s,t}\right)=-\frac{1}{s\left(2n-1\right)}\langle v|\sum_{j=1}^{2n-1}\left(\sum_{k=1}^{s}|\varphi^{k}\rangle_{j,j+1}\langle\varphi^{k}|\right)|u\rangle\sqrt{\frac{\pi\left(\mathfrak{t}\right)}{\pi\left(\mathfrak{s}\right)}}.\label{eq:P_jj1}
\end{equation}
In the foregoing equation for any summand $\sum_{k=1}^{s}|\varphi^{k}\rangle_{j,j+1}\langle\varphi^{k}|$, if the two strings $u$ and $v$ differ in any position other than
 $j,j+1$, then $P\left(\mathfrak{s,t}\right)=0$. Therefore, $P\left(\mathfrak{s,t}\right)$
is only nonzero when $\mathfrak{s}$ can be obtained from $\mathfrak{t}$
by single insertion or removal of a peak (i.e., $00\leftrightarrow u^{k} d^{k}$)
or vice versa. Next, we evaluate $P\left(\mathfrak{s,s}\right)$
\begin{eqnarray*}
P\left(\mathfrak{s,s}\right) & = & 1-\frac{1}{s\left(2n-1\right)}\langle\mathfrak{s}|H_{eff}|\mathfrak{s}\rangle = 1-\frac{\left(\begin{array}{c}
2n\\
2m
\end{array}\right)^{-1}}{s\left(2n-1\right)}\sum_{\mbox{Dyck}\left(u\right)=\mathfrak{s}}\sum_{\mbox{Dyck}\left(v\right)=\mathfrak{s}}\langle v|H_{int}|u\rangle
\end{eqnarray*}
but $\langle v|H_{int}|u\rangle=\frac{s}{2}$ for local moves that
take $00\leftrightarrow00$ and $\langle v|H_{int}|u\rangle=\frac{1}{2}$
for moves $ u^{k} d^{k}\leftrightarrow u^{k} d^{k}$. In Eq. \ref{eq:P_jj1},
we have $\langle\mathfrak{s}|H_{eff}|\mathfrak{s}\rangle\le\frac{\left(2n-1\right)s}{2}$,
hence $P\left(\mathfrak{s},\mathfrak{s}\right)\ge1/2$.

Next consider $\mathfrak{t}\ne\mathfrak{s}$. If $\mathfrak{s}\in D_{m}^{s}$
then $\mathfrak{t}\in D_{m\pm1}^{s}$, which is obtained from $\mathfrak{s}$
by removing or inserting a peak of any color (i.e., $00\leftrightarrow u^{k} d^{k}$).
Using the definitions in Eqs. (\ref{eq:V} and \ref{eq:P})
\begin{eqnarray*}
P\left(\mathfrak{s,t}\right)=-\frac{1}{s\left(2n-1\right)}\sqrt{\left(\begin{array}{c}
2n\\
2m
\end{array}\right)^{-1}\left(\begin{array}{c}
2n\\
2\left(m\pm1\right)
\end{array}\right)^{-1}} \sum_{\mbox{Dyck}\left(u\right)=\mathfrak{s}}\sum_{\mbox{Dyck}\left(v\right)=\mathfrak{t}}\langle u|H_{int}|v\rangle\sqrt{\frac{\pi\left(\mathfrak{t}\right)}{\pi\left(\mathfrak{s}\right)}}
\end{eqnarray*}
where $\sqrt{\frac{\pi\left(\mathfrak{t}\right)}{\pi\left(\mathfrak{s}\right)}}=\left|\frac{\langle\mathfrak{t}|D^{s}\rangle}{\langle\mathfrak{s}|D^{s}\rangle}\right|=\sqrt{\left(\begin{array}{c}
2n\\
2\left(m\pm1\right)
\end{array}\right)\left(\begin{array}{c}
2n\\
2m
\end{array}\right)^{-1}}$, giving 
\begin{equation}
P\left(\mathfrak{s,t}\right)=-\frac{1}{s\left(2n-1\right)}\left(\begin{array}{c}
2n\\
2m
\end{array}\right)^{-1}\sum_{\mbox{Dyck}\left(u\right)=\mathfrak{s}}\sum_{\mbox{Dyck}\left(v\right)=\mathfrak{t}}\langle u|H_{int}|v\rangle,\label{eq:PLemProof}
\end{equation}

First suppose $\mathfrak{t}\in D_{m+1}^{s}$. Let us fix some $j\in\left[0,2m\right]$
such that $\mathfrak{t}$ can be obtained from $\mathfrak{s}$ by
inserting a pair $ u^{k} d^{k}$ of a given color $k$, between $\mathfrak{s}_{j}$
and $\mathfrak{s}_{j+1}$. For any string $u$ such that $\mbox{Dyck}\left[u\right]=\mathfrak{s}$
in which $\mathfrak{s}_{j}$ and $\mathfrak{s}_{j+1}$ are separated
by at least two zeros one can find at least one $v$ with $\mbox{Dyck}\left(v\right)=\mathfrak{t}$
such that $\langle u|H_{int}|v\rangle=-\frac{1}{2}$ . The fraction
of strings that are obtained from randomly inserting two consecutive
zeros into a string of length $2n$ are at least $\frac{1}{4n^{2}}$,
which is also a lower bound for inserting $2\left(n-m\right)$ zeros
into $\mathfrak{s}$. This combined with Eq. (\ref{eq:PLemProof})
and the fact that there are $s$ different peaks gives 
\[
P\left(\mathfrak{s,t}\right)\ge-\frac{1}{8n^{3}}\langle u|H_{int}|v\rangle=\frac{1}{16n^{3}}.
\]

Now suppose $\mathfrak{t}\in D_{m-1}^{s}$. Let us fix some $j\in\left[1,2m-1\right]$
such that $\mathfrak{t}$ can be obtained from $\mathfrak{s}$ by
removing the pair $\mathfrak{s}_{j}\mathfrak{s}_{j+1}= u^{k} d^{k}$
for some color $k$. The fraction of strings $u$ such that $\mbox{Dyck}\left[u\right]=\mathfrak{s}$
and that no zeros are inserted between $\mathfrak{s}_{j}$ and $\mathfrak{s}_{j+1}$
are at least $\frac{1}{2n}$. Similar to above we have
\[
P\left(\mathfrak{s,t}\right)\ge-\frac{1}{4n^{2}}\langle u|H_{int}|v\rangle=\frac{1}{8n^{2}}.
\]\end{proof}

Hence, to prove that the Hamiltonian has a $\mbox{poly}\left(1/n\right)$
gap, it suffices to prove that $P$ has a polynomial gap, $\left(1-\lambda_{2}\left(P\right)\right)\ge n^{-\mathcal{O}\left(1\right)}.$
We do so by proving that the Markov chain is rapidly mixing \cite{king_ConductanceSurvey2003}. 
\subsubsection{Rapidly mixing Markov chain: Canonical path technique}
One way to prove that the Markov chain has a large spectral gap, is
to show that it is rapidly mixing, or equivalently, it has a high
{\it conductivity} \cite{king_ConductanceSurvey2003}. Showing this ensures that
starting from any arbitrary Dyck walk, one can move along the edges
of the graph and ultimately reach any other Dyck walk quickly (i.e.,
in polynomial time). 

We prove that $P$ mixes rapidly and hence has a large gap using the
\textit{canonical path technique, }which ensures that there is a connected
path via which one can obtain any $\mathfrak{t}\in D_{m}^{s}$ from
any $\mathfrak{s}\in D_{k}^{s}$ by a sequence of insertion and removal
of peaks such than no intermediate edge is overloaded. Perhaps it
is helpful to give a traffic analogy that would illustrate the canonical
path technique. A city is rapidly mixing, or equivalently, has high
conductivity if it has a low traffic. We say the city has a low traffic,
if one can drive between any two arbitrary houses efficiently. One
way to ensure this, is to show that between any two arbitrary chosen
houses there are a sequence of roads that connect them such that none
of the roads is overly used by other drivers (i.e., none of which
is congested). Therefore, one never gets ``stuck'' in traffic in
any intermediate road and consequently reaches the destination quickly.

For the canonical path method, we specify a path $\gamma\left(\mathfrak{s},\mathfrak{t}\right)$
between two arbitrary states of the Markov chain. The canonical
path theorem shows that for a reversible Markov chain the spectral
gap is \cite{Sinclair1992}
\begin{equation}
1-\lambda_{2}\ge\frac{1}{\rho L}\label{eq:MarkovGapRaw}
\end{equation}
where the maximum edge load $\rho$ is 
\begin{equation}
\rho=\underset{\left(a,b\right)\in E}{\max}\mbox{ }\frac{1}{\pi\left(a\right)P\left(a,b\right)}\sum_{\left(a,b\right)\in\gamma_{\mathfrak{s},\mathfrak{t}}}\pi\left(\mathfrak{s}\right)\pi\left(\mathfrak{t}\right).\label{eq:MaxEdgeLoad}
\end{equation}
The probability distribution $\pi\left(\mathfrak{s}\right)$ is the
stationary distribution of the Markov chain, and $L=\max_{\left(\mathfrak{s,t}\right)}\left|\gamma_{\mathfrak{s},\mathfrak{t}}\right|$
is the length of the longest canonical path. Thus, if no edge is covered
by too many canonical paths, the Markov chain will mix rapidly. 

The transition matrix $P$ describes a random walk on the graph of
Dyck walks, where two walks $\mathfrak{s}\in D_{k}^{s}$ and $\mathfrak{t}\in D_{m}^{s}$$ $
are connected, i.e., have an edge between them, if $\mathfrak{t}$
can be obtained from $\mathfrak{s}$ by insertion/removal of a peak
of any color $ u^{k} d^{k}$.

After the proof of Lemma \ref{4s_fractionalMatching}, we shall prove
rapid mixing between any $\mathfrak{s}\in D_{k}^{s}$ and $\mathfrak{t}\in D_{m}^{s}$.
However, to better illustrate the method, for now let $\mathfrak{s}$
and $\mathfrak{t}$ be two Dyck walks of length $2n$, then the Markov
chain $P$ takes $\mathfrak{s}$ to $\mathfrak{t}$ via a sequence
of steps that essentially does the following:
\begin{enumerate}
\item Pick a position between $1$ and $2n-1$ on the Dyck path $\mathfrak{s}$
at random.
\item If there is a peak there, remove it to get a path of length $2n-2$.
Here, a peak is a coordinate on the path that is greater than both
of its neighbors.
\item Insert a peak at random position between $0$ and $2n-2$ with a color
randomly chosen out of the $s$ possibilities uniformly.
\end{enumerate}

As shown in Lemma \ref{lem:L2}, $P\left(a,b\right)\ge1/16n^{3}=\Omega\left(n^{-3}\right)$.
We need to use multi-edges in cases where cutting off two peaks, or
inserting a peak at two different positions, gives the same Dyck path.
The stationary distribution is uniform, so $\pi\left(a\right)=\frac{1}{C_{n}s^{n}}\approx\frac{\sqrt{\pi}n^{3/2}}{\left(4s\right)^{n}}$
where, as before, $C_{n}$ is the $n^{\mbox{th}}$ Catalan number
and $s^{n}$ corresponds to the $s$ possible colorings of any given
Dyck path. Finally in our canonical paths construction, each path
will be of maximum length $2n$, and no edge will appear in more than
$2n\left(4s\right)^{n-1}$ paths. 

Putting this together, we get that $\rho\le\frac{2\sqrt{\pi}}{s}n^{\frac{11}{2}}$
and the spectral gap is (see Eq. \ref{eq:MarkovGapRaw}) 
\begin{equation}
1-\lambda_{2}\ge\frac{s}{4\sqrt{\pi}n^{\frac{13}{2}}}\mbox{ }.\label{eq:Pgap}
\end{equation}
\begin{figure}
\centerline{\includegraphics[width=.45\textwidth]{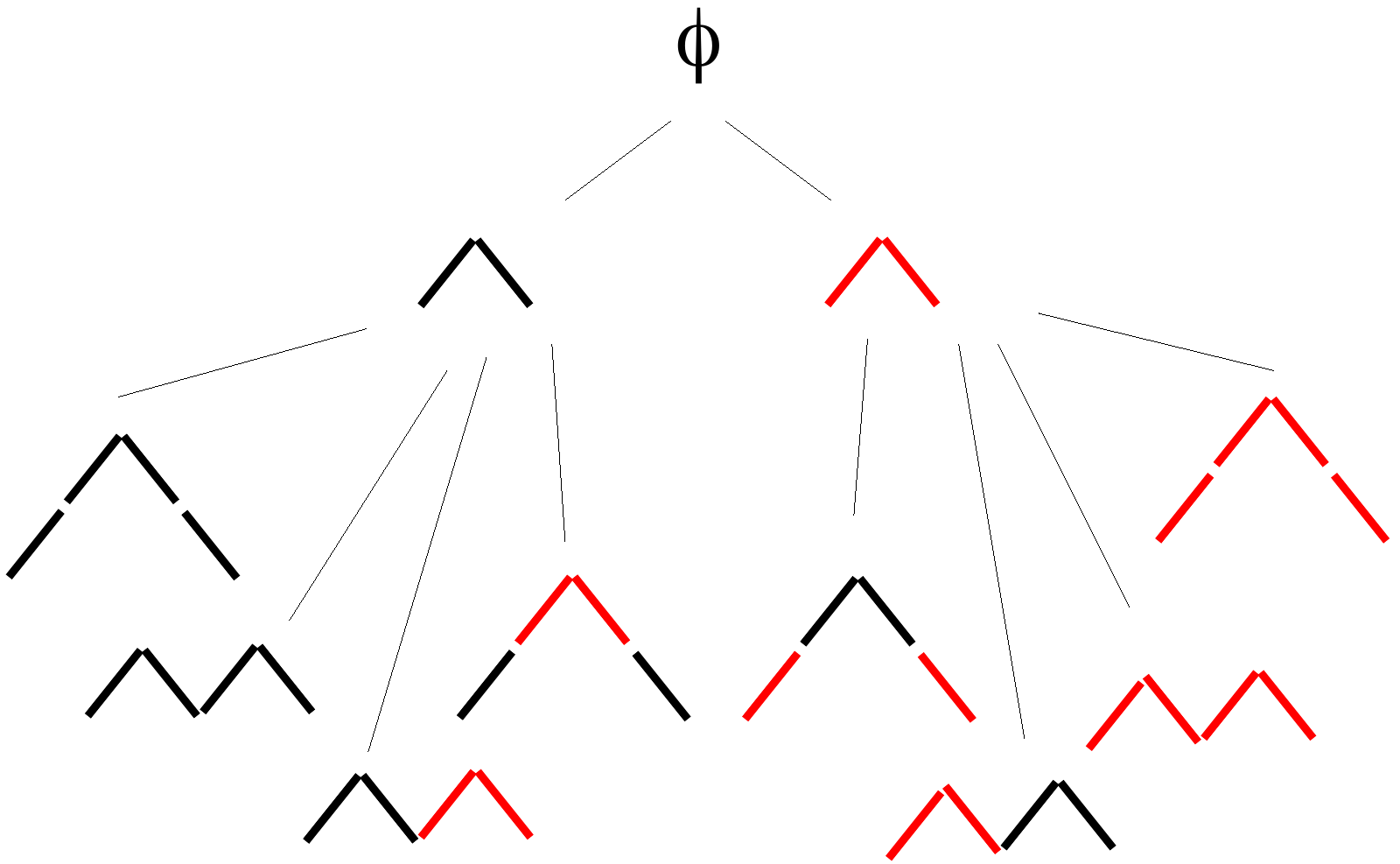}}
\caption{A tree containing all $2-$colored Dyck
walks of length 4 or less. }\label{fig:A-tree-containing}
\end{figure}

In detail, how can such a canonical path be constructed? Given two
$s-$colored Dyck paths, we define the canonical path between them
in the following way. We arrange all the Dyck paths of length $\le2n$
into a tree where the root of the tree is the empty Dyck path (of
length $0$), and where level $m$ contains all Dyck paths of length
$2m$. We will require that any node can be taken to its parent by
removing a peak from the Dyck path, and that no node has more than
$4s$ children. For example Fig. \ref{fig:A-tree-containing} gives
such a tree containing all Dyck paths of length $4$ with two colors.
\begin{figure}
\centerline{\includegraphics[width=.5\textwidth]{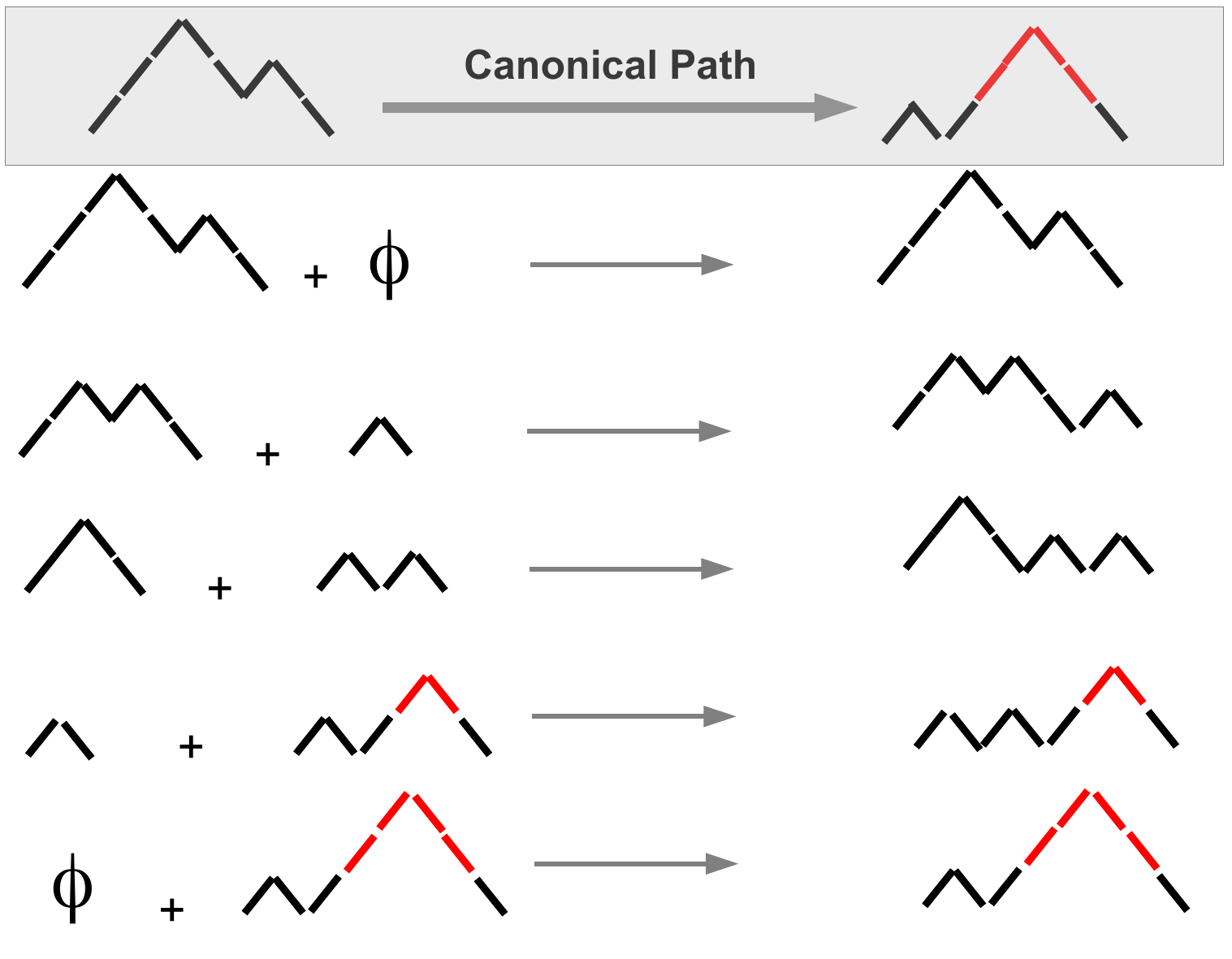}}
\caption{Canonical path between $2-$colored
Dyck walks of length $8$. By cutting peaks we shrink the walk on
the top left and by inserting peaks we grow the final walk (top right). }\label{fig:Canonical-path-between}
\end{figure}

Now suppose we wish to find the canonical path between two Dyck paths
$\mathfrak{s}$ and $\mathfrak{t}$, where these two Dyck paths have
length $2n$. By considering the path in the tree from the leaf $\mathfrak{s}$
to the root, we obtain a sequence of Dyck paths $\mathfrak{s}=\mathfrak{s}_{2n},\mathfrak{s}_{2n-2},\mathfrak{s}_{2n-4},\cdots,\mathfrak{s}_{0}=\emptyset$
and similarly for $\mathfrak{t}$. For the $m^{\mbox{th}}$ Dyck path
in our canonical path we use the concatenation of the two Dyck paths
$\mathfrak{s}_{2n-2m}$ and $\mathfrak{t}_{2m}$. For example, see
Fig. \ref{fig:Canonical-path-between} for an example canonical path
determined using the tree in Fig. \ref{fig:A-tree-containing}.

It is clear that the length of this canonical path is at most $2n$.
Suppose we have an edge $\left(a,b\right)$ on our random walk between
two Dyck paths. This edge could appear as the $m^{\mbox{th}}$  step
in a canonical path for $m=1,2,\cdots,n$. If it appears at step $m$,
then this transition corresponds to the transition between two Dyck
paths $\mathfrak{s}_{2n-2m+2}\mathfrak{t}_{2m-2}\rightarrow\mathfrak{s}_{2n-2m}\mathfrak{t}_{2m}$.
This edge will appear on any canonical path between a descendant of
$\mathfrak{s}_{2n-2m+2}$ and a descendant of $\mathfrak{t}_{2m}$
in our tree. Now, the node $\mathfrak{t}_{2m}$ has $\left(4s\right)^{n-m}$
descendant leaves in our tree, and $\mathfrak{s}_{2n-2m+2}$ has $4^{m-1}$
descendent leaves in our tree, so there are at most $\left(4s\right)^{n-1}$
different pairs $\mathfrak{s},\mathfrak{t}$ for which this transitions
is the $m^{\mbox{th}}$ step on the canonical path. Thus, the edge
$\left(a,b\right)$ lies on at most $2n\left(4s\right)^{n-1}$ canonical
paths. 

Now, the remaining step in our proof is building the tree. For this,
all we need to do is show that we can map the Dyck paths of length
$2n$ onto the Dyck paths of length $2n-2$ so that every Dyck path
of length $2n-2$ has at most $4s$ pre-images. This mapping will
define the edges between the nodes on level $n-1$ and level $n$
of our tree.
\begin{figure}[h]
\centering
{\includegraphics[width=0.4\textwidth]{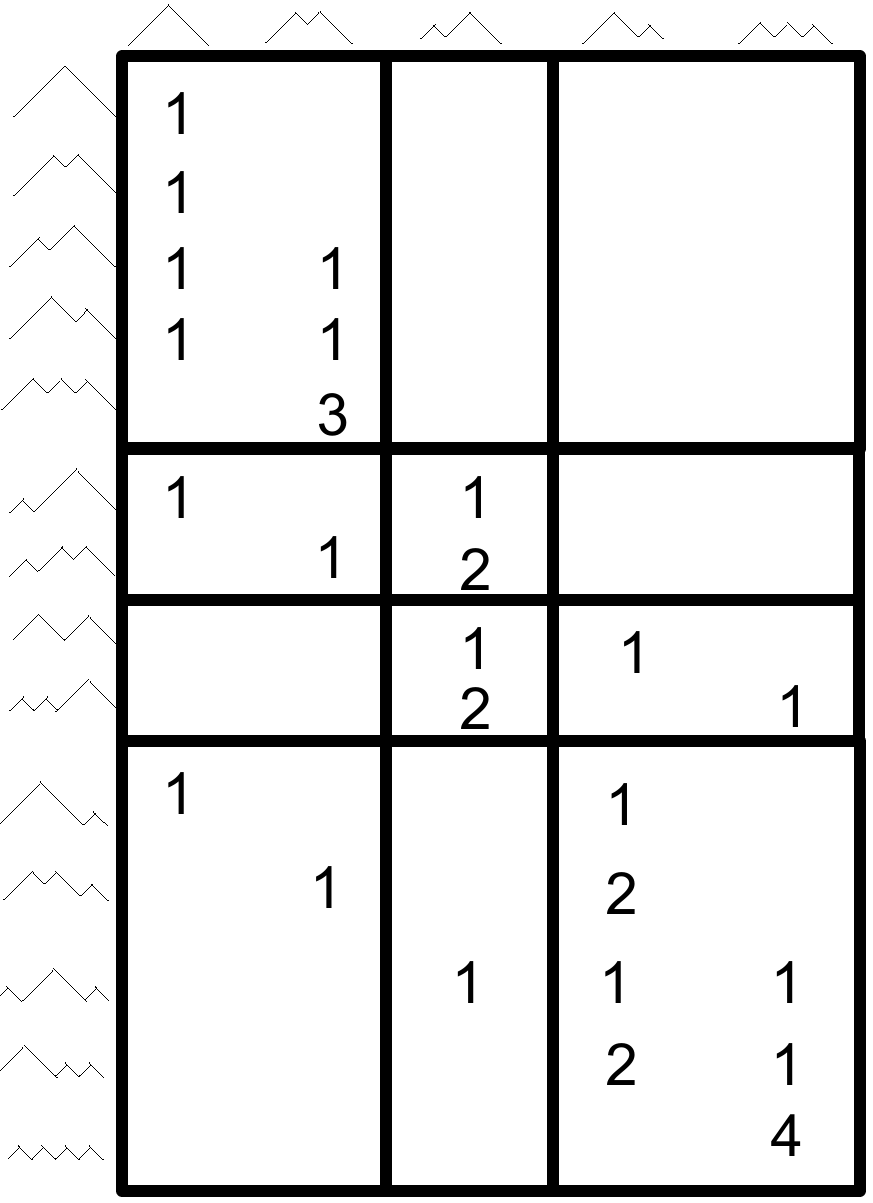}}
{\includegraphics[width=0.4\textwidth]{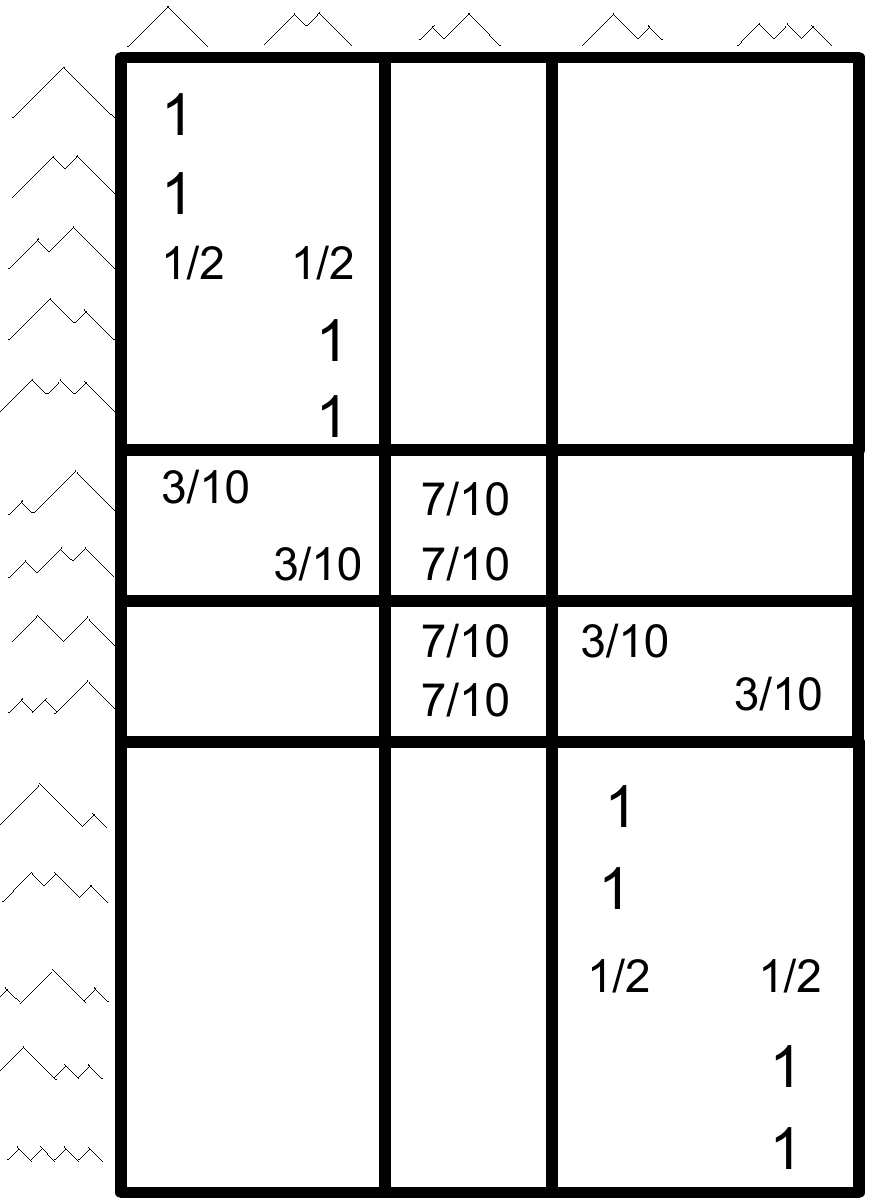}}
\caption{Fractional matching where for simplicity we set $s=1$. The matrices are $14\times5$}
\label{fig:Fractional-matching.-The}
\end{figure}

In order to build this tree, we use a\textit{ fractional matching
theorem} which says that if we can build a fractional matching between
paths of length $2n$ and paths of length $2n-2$, then we can find
a matching as follows. Consider a matrix $m_{ij}$ with the rows labeled
by the colored Dyck paths of length $2n$ and the columns labeled
by paths of length $2n-2$. Make $m_{ij}\ge1$ if column $j$ can
be obtained from row $i$ by removing a peak of a given color and
$0$ otherwise (see Fig. \ref{fig:Fractional-matching.-The} for examples).
We let $m_{ij}$ be the number of ways of getting from path $j$ to
path $i$ by removing a peak. Now the definition of fractional matching
is a matrix $x_{ij}$ such that $0\le x\le1$, $\sum_{i}x_{ij}\le4s$,
and $\sum_{j}x_{ij}=1$. We will build such matrix by induction. In
fact, we will show that we can build a matrix $x_{ij}$ where all
the column sums are equal and all the row sums are $1$. This additional
hypothesis will let us use induction to construct the fractional matching.

To construct the fractional matching, we will first put the Dyck paths
into a specific order. Recall the Catalan numbers are defined by a
recursion
\[C_{n}=\sum_{m=0}^{n-1}C_{m}C_{n-m-1}, \]
 where $C_{0}=1$. Translating this into Dyck paths, each path of
length $2n$ can be associated with a pair of paths, of length $2m$
and $2\left(n-m-1\right)$, where $0\le m\le n$. 
\begin{figure}
\centerline{\includegraphics[width=.7\textwidth]{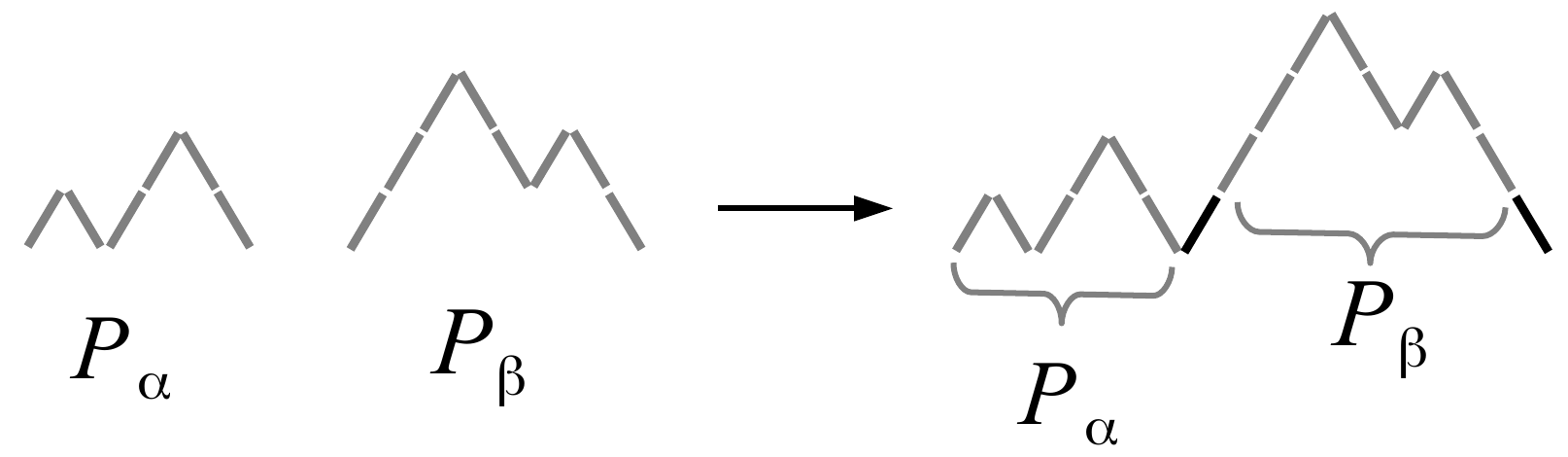}}
\caption{The result of the concatenation of two paths
$P_{\alpha}$ and $P_{\beta}$, where only for the sake of clarity
we made the steps up and down added before and after $P_{\beta}$
black. }\label{fig:ConcatinDyck}
\end{figure}

To take two paths, $P_{\alpha}$ of length $2m$ and $P_{\beta}$
of length $2\left(n-m-1\right)$ and obtain a path of length $2n$,
add a step up before $P_{\beta}$ and a step down after $P_{\beta}$,
and concatenate them (See Fig. \ref{fig:ConcatinDyck}). This is a
one-to-one correspondence between pairs of Dyck paths whose length
sum to $2n-2$ and Dyck paths of length $2n$. We have already shown
one direction of this mapping. This mapping is reversible because
the first path $P_{\alpha}$ ends at the last point the path hits
the $x-$axis before its end. 

The construction of colored Dyck walks have a natural correspondence
too; they can be defined by what we call the \textit{colored Catalan
numbers}
\begin{equation}
s^{n}C_{n}=s\sum_{m=0}^{n-1}s^{m}C_{m}\left(s^{n-m-1}C_{n-m-1}\right),\label{eq:ColorCatalanRecur}
\end{equation}
where the $s$ multiplying the sum is the number of ways the step
up and step down before and after $P_{\beta}$ can be colored.

Now suppose we have a path of length $2n$ corresponding to a pair
of Dyck paths $P_{\alpha}P_{\beta}$, where $2\alpha+2\beta=2n-2$.
Suppose $\beta\ne0$. Then, when we remove a peak of a given color,
we will either end up with a path corresponding to $P_{\alpha-1}P_{\beta}$
or $P_{\alpha}P_{\beta-1}$. If $\beta=0$, we can also remove the
last peak to end up with the path $P_{\alpha}$. 

Thus the matrix $m_{ij}$ breaks into block diagonal pattern, with
the columns divided into blocks containing paths of the form $P_{\alpha}P_{\beta}$
where $\alpha+\beta=n-2$, and the rows divided into blocks of the
form $P_{\alpha}P_{\beta}$ where $\alpha+\beta=n-1$. Except for
the identity matrix added to the column $C_{n-1}$ rows, for each
column block there are only two non-zero row blocks, and vice versa.
In our construction, we never use the fact that there is an identity
added, so we will ignore the existence of this in the following (see
Fig. \ref{fig:Fractional-matching.-The}).

Let us look at these blocks more closely. The block of rows $P_{\alpha}P_{\beta}$
has $0$'s except in column blocks $P_{\alpha-1}P_{\beta}$ or $P_{\alpha}P_{\beta-1}$.
The sub-matrix $P_{\alpha}P_{\beta}\times P_{\alpha-1}P_{\beta}$
is simply the matrix $M_{\alpha}\otimes I_{C_{\beta}}$, where $M_{\alpha}$
is the matrix relating paths of length $2\alpha$ and $2\alpha-2$
and $I_{C_{\beta}}$ is an identity matrix of size $C_{\beta}$ ($\beta^{\mbox{th}}$
Catalan number). We assume by induction that we have a fractional
matching on $M_{\alpha}$ and $M_{\beta}$. By taking the tensor product
of these and an identity matrix, we obtain a fractional matching on
these sub-blocks (see Fig. \ref{fig:Fractional-matching.-The}). 

Now, we can construct the fractional matching by multiplying the fractional
matching for a sub-block $P_{\alpha}P_{\beta}\times P_{\alpha-1}P_{\beta}$
and $P_{\alpha}P_{\beta}\times P_{\alpha}P_{\beta-1}$ by appropriate
scalars, so that all the rows add to $1$ and all the columns have
the same sum. The column sum is $C_{n}s^{n}/\left(C_{n-1}s^{n-1}\right)<4s$.

How can we prove this super-tree exists? The proof is based on fractional
matching theorem and number of other results in linear programming
and have been spelled out in \cite{Movassagh2012_brackets}. 

In particular, we proved the following useful lemma, where we only
had one color (spin $s=1$)\cite{Movassagh2012_brackets}
\begin{lem}
\label{Bravyi-et-al}(Bravyi et al \cite{Movassagh2012_brackets})
Let $D_{m}$ be the set of Dyck paths of length $2m$. For any $m\ge1$
there exists a map $f:D_{m}\rightarrow D_{m-1}$ such that (i) the
image of any path $\mathfrak{s}\in D_{m}$ can be obtained from $\mathfrak{s}$
by removing a single $u d$ pair, (ii) any path $\mathfrak{t}\in D_{m-1}$
has at least one pre-image in $D_{m}$, and (iii) any path $\mathfrak{t}\in D_{m-1}$
has at most four pre-images in $D_{m}$.
\end{lem}
This lemma allows us to grow arbitrary long Dyck paths starting from
the empty string. Clearly at every level $m$, we can color each $\mathfrak{s}\in D_{m}$
 walk $s^{m}$ different ways, whereby we have obtain all the $s-$colored
Dyck walks of size $2m$ (i.e., the set $D_{m}^{s}$). We can similarly
obtain $D_{m-1}^{s}$ by all $s$-colorings of every $\mathfrak{t}\in D_{m-1}$;
each $\mathfrak{t}$ can be colored $s^{m-1}$ different ways. 

\begin{lem}
\label{4s_fractionalMatching}Let $D_{m}^{s}$ be the set of $s-$colored
Dyck paths of length $2m$. For any $m\ge1$ there exists a map $f:D_{m}^{s}\rightarrow D_{m-1}^{s}$
such that (i) the image of any path $\mathfrak{s}\in D_{m}^{s}$ can
be obtained from $\mathfrak{s}$ by removing a single $ u^{k} d^{k}$,
(ii) any path $\mathfrak{t}\in D_{m-1}^{s}$ has at least $s$ pre-images
in $D_{m}^{s}$, and (iii) any path $\mathfrak{t}\in D_{m-1}^{s}$
has at most $4s$ pre-images in $D_{m}^{s}$.
\end{lem}
\begin{proof}

On the level $m-1$ there are $s^{m-1}$ copies of any Dyck walk of
length $2m-2$, each with a unique coloring assignment. Similarly
at the level $m$ there are $s^{m}$ copies of any Dyck walk of length
$2m$ each with a unique coloring. For any fixed coloring at the $m-1^{\mbox{st }}$
level and a fixed choice of the color for $ u^{k} d^{k}$ the problem
reduces to Lemma \ref{Bravyi-et-al}, so (i) is satisfied. Similarly
for (ii), there is a pre-image for any choice of $ u^{k} d^{k}$
so there are at least $s$ such pre-images. To prove (iii) we note
that for every choice of coloring of the Dyck walks at the level $m-1$
and a fixed choice of color for $ u^{k} d^{k}$, the problem is identical
to the previous case and there are at most $4$ pre-images (Lemma
\ref{Bravyi-et-al}). Since there are $s$ choices to color the peak
that we remove, there are at most $4s$ pre-images in total. 
\end{proof}

With these preliminaries, we now return to the proof of rapid mixing
time of $P$, whereby we need to prove that the maximum edge load
$\rho$ between any two arbitrary paths $\mathfrak{s}\in D_{m}^{s}$
and $\mathfrak{t}\in D_{k}^{s}$ is $n^{\mathcal{O}\left(1\right)}$
which proves $1-\lambda_{2}\left(P\right)\ge n^{-\mathcal{O}\left(1\right)}.$

We define the canonical path $\gamma\left(\mathfrak{s},\mathfrak{t}\right)$
such that any intermediate state is the concatenation of two walks
$pq$ where $p\in D_{\ell'}$ is an ancestor of $\mathfrak{s}$ in
the super-tree and $q\in D_{\ell"}$ is an ancestor of $\mathfrak{t}$.
The canonical path starts with $p=\mathfrak{s}$, $q=\emptyset$ and
alternates between shrinking $p$ by taking steps towards the root
and growing $q$ by taking steps away from the root similar to what
was discussed above. The path terminates as soon as $p=\emptyset$
and $q=\mathfrak{t}$. If at some intermediate state $p=\emptyset$
then the subsequent shrinking steps are skipped over while if in some
intermediate step $q=\mathfrak{t}$ then the subsequent growing steps
are skipped over. At any intermediate step the length of the concatenated
walk $\left|pq\right|$ obeys
\begin{equation}
\min\left(\left|\mathfrak{s}\right|,\left|\mathfrak{t}\right|\right)\le\left|pq\right|\le\max\left(\left|\mathfrak{s}\right|,\left|\mathfrak{t}\right|\right).\label{eq:length_uv}
\end{equation}
Since any $\gamma\left(\mathfrak{s},\mathfrak{t}\right)$ has a length
that is at most $2n$, it is enough to bound $\rho$. Let the edge
with maximum load, denoted by $\rho\left(m,k,\ell',\ell"\right)$,
be between $a=pq$ to $b$, where as before 
\[
\mathfrak{s}\in D_{m}^{s},\quad t\in D_{k}^{s},\quad p\in D_{\ell'}^{s},\quad q\in D_{\ell"}^{s}\quad.
\]
For the sake of concreteness let $b$ be obtained from $a$ by growing
$q$ and shrinking $p$ (the other case is analogous). From Lemma
\ref{4s_fractionalMatching}, the number of possible descendent strings
$\mathfrak{s}$ from which $p$ is obtained by shrinking is at most
$\left(4s\right)^{m-\ell'}$. The number of possible ancestors of
$\mathfrak{t}$ is at most $\left(4s\right)^{k-\ell"}$. Since $\pi\left(\mathfrak{s}\right)=\pi\left(\mathfrak{t}\right)$
for all $\mathfrak{s}\in D_{m}^{s}$ and $\mathfrak{t}\in D_{k}^{s}$,
\begin{equation}
\rho\left(m,k,\ell',\ell"\right)\le\frac{\left(4s\right)^{m}\pi\left(\mathfrak{s}\right)\mbox{ }\left(4s\right)^{k}\pi\left(\mathfrak{t}\right)}{\left(4s\right)^{\ell'+\ell"}\pi\left(a\right)P\left(a,b\right)}\quad.\label{eq:rho_mkl}
\end{equation}
where by definition $\pi\left(w\right)=\langle w|D^{s}\rangle^{2}$
and using Eq. \ref{eq:groundStateHeff} we obtain  $\pi\left(w\right)=\left(\begin{array}{c}
2n\\
2w
\end{array}\right)/M_{2n,s}$. From Lemma \ref{lem:L2} we have $P\left(a,b\right)\ge1/16n^{3}$.
To bound the right hand side of inequality (\ref{eq:rho_mkl}), we
first prove that 
\begin{equation}
\left(4s\right)^{w}\pi\left(w\right)=\frac{\sigma_{w}}{\sqrt{\pi}w^{3/2}}\label{eq:sigma_fraction}
\end{equation}
where $\sigma_{w}\le1$ is the fraction of $s-$colored Motzkin paths
of length $w$. Indeed $\sigma_{m}=s^{w}C_{w}\left(\begin{array}{c}
2n\\
2w
\end{array}\right)/M_{2n,s}$, where $s^{w}$ is the number of colorings of the Dyck walks of length
$2w$ counted by the Catalan number $C_{w}$. Since $s^{w}C_{w}\approx\left(4s\right)^{w}/\sqrt{\pi}w^{3/2}$,
Eq. \ref{eq:sigma_fraction} holds. Hence, we have 
\[
\rho\left(m,k,\ell',\ell"\right)\le\frac{16}{\sqrt{\pi}}n^{3}\left(\frac{\ell'+\ell"}{m\mbox{ }k}\right)^{3/2}\frac{\sigma_{m}\sigma_{k}}{\sigma_{\ell'+\ell"}}\quad.
\]
Since $\left(\frac{\ell'+\ell"}{m\mbox{ }k}\right)$ is at most a
polynomial in $n$, it remains to bound $\frac{\sigma_{m}\sigma_{k}}{\sigma_{\ell'+\ell"}}$.
We comment that the maximum edge load is always at least one (for
a fully connected graph). As mentioned above, in the canonical path,
we add a polynomial number of terms so it suffices to prove that the
ratio $\frac{\sigma_{m}\sigma_{k}}{\sigma_{\ell'+\ell"}}$ is small;
indeed 
\begin{eqnarray*}
\frac{\sigma_{m}\sigma_{k}}{\sigma_{\ell'+\ell"}} & = & \frac{1}{M_{2n,s}}\frac{s^{m}C_{m}\left(\begin{array}{c}
2n\\
2m
\end{array}\right)s^{k}C_{k}\left(\begin{array}{c}
2n\\
2k
\end{array}\right)}{s^{\ell'+\ell"}C_{\ell'+\ell"}\left(\begin{array}{c}
2n\\
2\left(\ell'+\ell"\right)
\end{array}\right)}\quad.
\end{eqnarray*}
But $M_{2n,s}=\sum_{w=1}^{n}s^{w}C_{w}\left(\begin{array}{c}
2n\\
2w
\end{array}\right)$, which includes terms with $w=m$ and $w=k$ so we have
\[
\frac{\sigma_{m}\sigma_{k}}{\sigma_{\ell'+\ell"}}\le\frac{1}{s^{\ell'+\ell"}C_{\ell'+\ell"}\left(\begin{array}{c}
2n\\
2\left(\ell'+\ell"\right)
\end{array}\right)}\le1\quad.
\]

We conclude that $\rho\le n^{\mathcal{O}\left(1\right)}$, which implies
that the spectral gap $1-\lambda_{2}\left(P\right)\ge n^{-\mathcal{O}\left(1\right)}$.
This completes the $\mbox{poly}\left(1/n\right)$ proof of the lower
bound for the gap in the balanced subspace.
\subsection{Smallest energy of unbalanced and/or crossed states: $\mbox{poly}\left(1/n\right)$
lower bound}
Previously we proved that if we restrict the Hamiltonian to the space
where there are an excess number of right or step up then
the smallest eigenvalue is indeed lower bounded by a polynomial in
$1/n$. The problem at hand is different for there are different types
of steps and in addition there is the possibility of having
mismatches where $\Pi^{cross}$ does not vanish. 

To establish the gap to be a polynomial in $1/n$ we need to lower
bound the ground state energy of the Hamiltonian in the unbalanced
subspace, where for example a sub string configuration such as $u^1 u^2 d^1 d^2$
can occur. 

It is sufficient to separately prove lower bounds on: 1. the subspace
with only mismatches 2. imbalance subspace without any mismatch. The
reason for the sufficiency is that including mismatches to an imbalance
space or vice versa can only increase the energy. 

\textit{Pure mismatch:} In this case the energy penalties come from
$\sum_{i}\Pi_{i,i+1}^{cross}$. Let us assume there is a single mismatch
such as $g_{0} \mbox{ }u^1 g_{1}\mbox{ }u^2 \mbox{ }g_{2}\mbox{ } d^1 \mbox{ }g_{3}\mbox{ } d^2 \mbox{ }g_{4}$,
where $g_{1},\dots,g_{4}$ are strings in the alphabet $\left\{ 0, u^{1},\dots, u^{s}, d^{1},\dots, d^{s}\right\} $
such that if we ignore the mismatching, the string $g_{0} \mbox{ }u^1 g_{1}\mbox{ }u^2 \mbox{ }g_{2}\mbox{ } d^1 \mbox{ }g_{3}\mbox{ } d^2 \mbox{ }g_{4}$
would indeed be a colored-Motzkin walk. Moreover, we can assume there
is only a single mismatch as in the example just given since having
more mismatches results in more penalties and can only increase the
energy. Recall that the Hamiltonian is 
\[
H=\sum_{j=1}^{2n-1}\Pi_{j,j+1}+\sum_{j=1}^{2n-1}\Pi_{j,j+1}^{cross}
\]
where we can ignore the boundary terms as we are restricting ourselves
to only mismatched subspaces. $\Pi_{j,j+1}$ is the hopping Hamiltonian
that allows the propagation of steps of any type through the
vacuum. In $g_{0} \mbox{ }u^1 g_{1}\mbox{ }u^2 \mbox{ }g_{2}\mbox{ } d^1 \mbox{ }g_{3}\mbox{ } d^2 \mbox{ }g_{4}$
suppose $u^2$ (appearing after $g_{1}$) is at site $i$
and the first step down appearing after $g_{2}$ is at site
$j$ and $d^2$ between $g_{3}$ and $g_{4}$ is on site
$k$. If we take the hopping amplitude on sites $i$ and $k$ to be
zero then the energy can only decrease and the problem reduces to
the case where there is a chain of length $k-i$ with a single excess
step down at site $j$. So the problem formally reduces to
the previous problem \cite{Movassagh2012_brackets} on a chain of
length $k-i$. Therefore the previous polynomial lower bound also
lower bounds this case. 
\textit{Imbalance subspace without a mismatch:} The Hamiltonian now
reads
\[H=\Pi_{boundary}+\sum_{j=1}^{2n-1}\Pi_{j,j+1}\]
where $\sum_{j=1}^{2n-1}\Pi_{j,j+1}^{cross}$ vanishes and therefore
can be ignored. We need to lower bound the smallest eigenvalue on
strings of type \[u_{0} d^{i_{1}}u_{1} d^{i_{2}}u_{2}\cdots  d^{i_{k}}u_{k} u^{j_{m}}\cdots v_{2} u^{j_{2}}v_{1} u^{j_{1}}v_{0}\]
where $u_{i}$ and $v_{i}$ are $s$-colored Motzkin walks and $i_{p},j_{q}$
can take on any values in $\left\{ 1,2,\cdots,s\right\} $. Since
the spectrum of $H$ in this subspace only depends on the total number
of excess up and down steps, we can focus on having only step down
imbalanced walks, whereby we simplify the analysis and drop the boundary
terms $\sum_{i=1}^{s}\mbox{ } |u^{i}\rangle_{2n} \langle u^{i}|$ as doing
so can only decrease the energy. Below we use a similar argument as
before \cite{Movassagh2012_brackets}. Given any string $g$ in the
imbalanced subspace with only excess step down, let $\tilde{u}\in\left\{ 0, u^{1},\dots, u^{s}, d^{1},\dots, d^{s},x,y\right\} $
be the string obtained from $g$ by i) replace the first unmatched
step down by $x$ and ii) replace all other unmatched steps
in $g$ by $y$. We can define a new Hilbert space $\tilde{{\cal H}}$
whose basis vectors are $|\tilde{g}\rangle$. Consider a Hamiltonian
\[
\tilde{H}=|x\rangle_{1}\langle x|+\sum_{j=1}^{2n-1}\Pi_{j,j+1}+\Theta_{j,j+1}^{x}+\Theta_{j,j+1}^{y}
\]
where $\Theta^{x}$ and $\Theta^{y}$ are projectors onto the states
$|0x\rangle-|x0\rangle$ and $|0y\rangle-|y0\rangle$ respectively (with
proper normalizations). Since $\langle u|H|v\rangle=\langle\tilde{u}|\tilde{H}|\tilde{v}\rangle$
for any $u,v$ the spectrum of $H$ and $\tilde{H}$ coincide in this
subspace. We can further drop $\Theta^{y}$ terms as doing so only
decreases the energy. Therefore, it is sufficient to consider the
simplified Hamiltonian 
\[
H^{x}=|x\rangle_{1}\langle x|+\sum_{j=1}^{2n-1}\Pi_{j,j+1}+\Theta_{j,j+1}^{x}
\]
which act on $\tilde{{\cal H}}$ and position of $y$ particles are
constants of motion of $H^{x}$. An entirely a similar argument as
in \cite{Movassagh2012_brackets} shows that we can only analyze the
interval between $1$ and the first $y-$particle, whereby the relevant
Hilbert space becomes the span of (as before we denote the set of
Motzkin paths of length $k$ by ${\cal M}_{k}$) 
\[
|u\rangle\otimes |x\rangle\otimes |v\rangle,\quad\mbox{where }u\in{\cal M}_{j-1},\quad v\in{\cal M}_{2n-j}.
\]

To use the projection lemma define 
\[
H_{\epsilon}^{x}=\sum_{j=1}^{2n-1}\Pi_{j,j+1}+\epsilon\left\{ |x\rangle_{1}\langle x|+\sum_{j=1}^{2n-1}\Theta_{j,j+1}^{x}\right\} 
\]
and an effective Hopping Hamiltonian $H_{eff}$ can be defined whose
ground state lower bounds the ground state of $H_{\epsilon}^{x}\le H^{x}$.
$H_{eff}$ is defined by
\[
H_{eff}=|1\rangle\langle1|+\sum_{j=1}^{2n-1}\Gamma_{j,j+1}
\]
where 
\begin{align*}
\Gamma_{j,j+1} & =\alpha_{j}^{2}\mbox{ }|j\rangle\langle j|\mbox{ }+\beta_{j}^{2}\mbox{ }|j+1\rangle\langle j+1|\\
 & -\alpha_{j}\beta_{j}\left\{ \mbox{ }|j\rangle\langle j+1|\mbox{ }+\mbox{ }|j+1\rangle\langle j|\mbox{ }\right\} 
\end{align*}
is a rank-1 projector. The coefficients are now different from the
previous case and are given by 
\begin{eqnarray*}
\alpha_{j}^{2} & \equiv & \langle\psi_{j}|\Theta_{j,j+1}^{x}|\psi_{j}\rangle=\frac{M_{2n-j-1}}{2s\mbox{ }M_{2n-j}}\\
\beta_{j}^{2} & \equiv & \langle\psi_{j+1}|\Theta_{j,j+1}^{x}|\psi_{j+1}\rangle=\frac{M_{j-1}}{2s\mbox{ }M_{j}}
\end{eqnarray*}
and lastly 
\[
-\alpha_{j}\beta_{j}=-\frac{1}{2s}\sqrt{\frac{M_{2n-j-1}}{M_{2n-j}}\frac{M_{j-1}}{M_{j}}}
\]
where $M_{k}$ is the $k^{\mbox{th}}$ Motzkin number which is the
number of Motzkin walks in $k$ steps. Applying the projection lemma
we have $\lambda_{1}\left(H_{\epsilon}^{x}\right)\ge\epsilon\lambda_{1}\left(H_{eff}\right)$
and it suffices to show that $\lambda_{1}\left(H_{eff}\right)\ge n^{-\mathcal{O}\left(1\right)}$. 

The hopping Hamiltonian without the ``repulsive potential'' $|1\rangle\langle1|$
is 
\[
H_{move}\equiv\sum_{j=1}^{2n-1}\Gamma_{j,j+1}\quad.
\]
This is a FF Hamiltonian with the unique ground state 
\begin{equation}
|g\rangle\sim\sum_{j=1}^{2n}s^{n-\frac{1}{2}}\sqrt{M_{j-1}M_{2n-j}}\mbox{ }|j\rangle\quad.\label{eq:g_State}
\end{equation}
As before we bound the spectral gap of $H_{move}$ and use the Projection
Lemma to lower bound. Let $\pi\left(j\right)=\langle j|g\rangle^{2}$.
For any $a,b\in\left[1,2n\right]$ we define 
\[
P\left(j,k\right)=\delta_{j,k}-\langle j|H_{move}|k\rangle\sqrt{\frac{\pi\left(k\right)}{\pi\left(j\right)}}
\]
and a simple algebra shows that 
\[
P\left(j,j+1\right)=\frac{M_{2n-j-1}}{2sM_{2n-j}}\qquad\mbox{ and }P\left(j+1,j\right)=\frac{M_{j-1}}{2sM_{j}}
\]
are the only off diagonal matrix elements of $P$. Using Lemma 7 in
\cite{Movassagh2012_brackets} that shows $\frac{1}{3}\le\frac{M_{k}}{M_{k+1}}\le1$
we conclude that 
\[
\frac{1}{6s}\le P\left(j,j\pm1\right)\le\frac{1}{2s}\quad\forall j\quad.
\]
Consequently the diagonal elements of $P$ are non-negative and it
can be considered as a transition matrix. Moreover, using Eq. \ref{eq:g_State}
we conclude that 
\[
n^{-\mathcal{O}\left(1\right)}\le\frac{\pi\left(k\right)}{\pi\left(j\right)}\le n^{\mathcal{O}\left(1\right)}\qquad\forall\quad1\le j,k\le2n.
\]

We have $\min_{j}\pi\left(j\right)\ge n^{-\mathcal{O}\left(1\right)}$.
This is sufficient to bound the spectral gap of $P$ as shown in \cite{Movassagh2012_brackets}.
For example using the canonical paths theorem we get $1-\lambda_{2}\left(P\right)\ge\frac{1}{\rho\ell}$
with a canonical path that simply moves the $x-$particles from $u$
to $v$. Since the denominator in the maximum edge load given by Eq.
\ref{eq:MaxEdgeLoad} is lower bounded by $n^{-\mathcal{O}\left(1\right)}$
we conclude that the gap of $P$ is polynomially lower bounded and
that $\lambda_{2}\left(H_{move}\right)\ge n^{-\mathcal{O}\left(1\right)}$.

Lastly, one can apply the Projection lemma to $H_{eff}$ by making
$|1\rangle\langle1|$ a perturbation. The effective first order Hamiltonian
will now be constant $\langle1|g\rangle^{2}=\pi\left(1\right)\ge n^{-\mathcal{O}\left(1\right)}$
which proves the bound $\lambda_{1}\left(H_{eff}\right)\ge n^{-\mathcal{O}\left(1\right)}$. 

\section{Presence of an external field}
The energy corrections obtained from first order degenerate perturbation
theory are $\Delta E_{m}$ as defined in the paper. Since only the
embedded Dyck walks in the Motzkin state couple to the external field
and contribute to the energy corrections, we need to count the number
of walks that start from zero and reach coordinates $\left(2n,m\right)$. 
\begin{rem}
For now, we pretend that the length of the chains is $n$ and \textit{not}
$2n$. At the end we multiply $n$ by a factor of $2$. 
\end{rem}

The number of walks of length $n$ with $s$ coloring that reach the
height $m$ (i.e., coordinate $\left(x,y\right)=\left(m,n\right)$)
is denoted here by $\Gamma\left(n,m\right)$. As before, $\Gamma\left(n,m\right)$
is counted by a refinement of the Ballot problem 
\begin{equation}
\Gamma\left(n,m\right)\equiv s^{m}\sum_{k=0}^{n-m}\left(\begin{array}{c}
n\\
k
\end{array}\right)s^{\frac{n-k-m}{2}}B_{n-k,m}\equiv s^{m}M_{n,m,s}\label{eq:Mm-1}
\end{equation}
where there are $\left(\begin{array}{c}
n\\
k
\end{array}\right)$ ways of putting $k$ zeros, $B_{n-k,m}$ is the solution of the Ballot
problem with height $m$ on $n-k$ walks (number of ``Dyck'' walks
on $n-k$ steps that end at height $m$), $s^{m}$ ways of coloring
the unmatched steps and $s^{\frac{n-k-m}{2}}$ ways of coloring
the matched ones. 

The energy corrections obtained from first order degenerate perturbation
theory are 
\[
\Delta E_{m}=\frac{\epsilon}{n}\langle g_{m}|F|g_{m}\rangle
\]
and 
\[
\Delta E_{m}=\frac{\epsilon}{n N_{m}}\sum_{i,k}\langle g_{m}^{i}|F|g_{m}^{k}\rangle
\]
where, $N_{m}$ is the total number of walks that start at coordinates
$\left(0,0\right)$ and end at $\left(2n,m\right)$ and 
\begin{eqnarray*}
|g_{m}\rangle & \equiv & \frac{1}{\sqrt{N_{m}}}\sum_{i}\mbox{ }|g_{m}^{i}\rangle\\
 & = & \frac{1}{\sqrt{N_{m}}}\sum_{i}|\mbox{state \ensuremath{i}\ensuremath{\mbox{ }}with }m\mbox{ extra left parenth.}\rangle
\end{eqnarray*}
It is clear that $0<\Delta E_{m}\le\epsilon/n$. Since only the embedded
Dyck walks couple to the external field and give positive energy contribution,
we have
\begin{eqnarray}
\langle g_{m}|F|g_{m}\rangle & = & \frac{s^{m}\sum_{i\ge0}\left(m+2i\right)M_{n,m,s,i}}{\Gamma\left(n,m\right)} \nonumber \\
& = & \frac{\sum_{i\ge0}\left(m+2i\right)M_{n,m,s,i}}{\sum_{i\ge0}M_{n,m,s,i}}\label{eq:expected_f_m}\\
 & = & m+2\frac{\sum_{i\ge0}i\mbox{ }M_{n,m,s,i}}{\sum_{i\ge0}M_{n,m,s,i}}\nonumber 
\end{eqnarray}
where $M_{n,m,s,i}$ is defined in Eq. \ref{eq:Trinomial} and $m+2i$
is the number of nonzero terms on the walk (i.e., $u$ and $d$
terms)-- there are $M_{n,m,s,i}$ of the walks and $s^{m}$ cancels. 
\begin{rem}
Another way to interpret this is that $\langle g_{m}|F|g_{m}\rangle$
is the expected length of lattice paths with only step up and down
reaching height $m$ embedded in colored Motzkin paths of length $n$,
where the expectation is taken with respect to a uniform measure over
all the walks with $m$ imbalances and $s$ colors.
\end{rem}

It is not hard to see that the saddle point of $\sum_{i\ge0}iM_{n,m,i,s}$
is equal to that of the numerator, which is given by Eq. \ref{eq:saddle}.
Eq. \ref{eq:expected_f_m} after replacing the sum over $i$ with
an integral over $\beta$, as we did in our entanglement entropy calculation
above, and extending to $\pm\infty$ becomes
\begin{eqnarray}
\langle g_{m}|F|g_{m}\rangle & = & 2\sigma n+\frac{m}{4\sqrt{s}}\left(\frac{m}{n}\right)+\frac{\left(4s-1\right)m}{64\mbox{ }s\sqrt{s}}\left(\frac{m}{n}\right)^{3}\nonumber \\
& + & 2\sqrt{n}\mbox{ }\frac{\int d\beta\mbox{ }\beta\exp\left(-\frac{\sqrt{s}}{\sigma^{2}}\beta^{2}\right)}{\int d\beta\mbox{ }\exp\left(-\frac{\sqrt{s}}{\sigma^{2}}\beta^{2}\right)}+\mathcal{O}\left(m\left(\frac{m}{n}\right)^{5}\right)\nonumber \\
 & \approx & 2\sigma n+\frac{m}{4\sqrt{s}}\left(\frac{m}{n}\right)+\frac{\left(4s-1\right)m}{64\mbox{ }s\sqrt{s}}\left(\frac{m}{n}\right)^{3}\label{eq:Fina_n}
\end{eqnarray}

Restoring the factor of $2$ the new energies induced by an external
field of what used to be zero energy states become
\begin{equation}
\frac{\epsilon}{n}\langle g_{m}|F|g_{m}\rangle=4\sigma\epsilon +\frac{\epsilon}{8\sqrt{s}}\left(\frac{m}{n}\right)^2+\frac{\left(4s-1\right)\epsilon }{512\mbox{ }s\sqrt{s}}\left(\frac{m}{n}\right)^{4}+\mathcal{O}\left(\epsilon m\left(\frac{m}{n}\right)^{5}\right)\label{eq:Final_2n}
\end{equation}
\section{Other Open problems}
\begin{enumerate}
\item Further investigation of the nature of the excited states. 
\item Proof of the $\mbox{poly}\left(1/n\right)$ gap for Hamiltonians with
interaction terms that create maximally entangled states out of the
vacuum, i.e., $|\varphi\rangle=\frac{1}{\sqrt{2}}\left\{ |00\rangle-\frac{1}{\sqrt{s}}\sum_{i=1}^{s} |u^{i} d^{i}\rangle\right\} $.
The present technique for proving lower bounds would fail as $P\left(\mathfrak{s,t}\right)$
can become negative.
\item Can a similar model for $d<5$ systems be constructed, where the gap
behaves similar to here and the entanglement entropy is long-ranged?
Previously, a fermionic $d=4$ model was proposed whose entanglement
entropy grows linearly with $n$ \cite{MovassaghThesis2012}. However,
we believe (not yet proved) that the gap is exponentially small for
that model. Can other models with $d<5$ be built such that the gap
closes slowly with $n$?
\item Is $\sqrt{n}$ entanglement entropy as much as one can get in 'physically
reasonable' models \cite{swingle2013universal}?
\item What does the continuum limit of the class of Hamiltonians proposed
here look like? 
\item It may be possible to improve the upper bound to be $\mathcal{O}\left(n^{-3}\right)$
for the model with boundaries.
\item We think the combinatorial techniques introduced here add to the toolbox
of methods for proving the gap of local Hamiltonians. It would be
interesting to see other applications of them.
\item The spin-spin correlation function in the ground state can in principle
be calculated using the techniques that were used to calculate entanglement
entropies. It would be interesting to know how the correlation functions
$\langle\sigma_{i}\sigma_{k}\rangle$ and $\langle\sigma_{i}\sigma_{i+1}\sigma_{k}\sigma_{k+1}\rangle$
scale with $\left|i-k\right|$.
\end{enumerate}

\end{document}